\documentclass[a4paper,UKenglish,cleveref,autoref,thm-restate,authorcolumns]{lipics-v2021}
\usepackage{amssymb,amsmath,amsthm}
\usepackage{graphicx}
\usepackage{booktabs}

\usepackage{color}
\usepackage{enumitem}
\usepackage{tikz}
\usepackage[linesnumbered,lined,boxed]{algorithm2e}
\usepackage{framed}
\usepackage{multirow}
\usepackage{tikz}
\usetikzlibrary{arrows.meta,calc}

\nolinenumbers

\hideLIPIcs
\makeatletter
\def\mailsymbol{}
\def\homesymbol{}
\def\orcidsymbol{}
\makeatother

\title{Near-Tight Approximation Algorithms for Bottleneck Multiple Knapsack Problems}

\titlerunning{Near-Tight Approximation for BMKP}

\author{Lin Chen}{Zhejiang University}{chenlin198662@zju.edu.cn}{}{}
\author{Tingwei Hu}{Zhejiang University}{tingweihu@zju.edu.cn}{}{}
\author{Yuchen Mao}{Zhejiang University}{maoyc@zju.edu.cn}{}{}
\author{Yong Chen}{Hangzhou Dianzi University}{chenyong@hdu.edu.cn}{}{}
\author{Lili Mei}{Hangzhou Dianzi University}{meilili@hdu.edu.cn}{}{}
\author{An Zhang}{Hangzhou Dianzi University}{anzhang@hdu.edu.cn}{}{}
\author{Guangting Chen}{Zhejiang University of Water Resources and Electric Power}{gtchen@hdu.edu.cn}{}{}
\author{Guochuan Zhang}{Zhejiang University}{zgc@zju.edu.cn}{}{}
\authorrunning{Lin Chen et al.}

\Copyright{Lin Chen, Tingwei Hu, Yuchen Mao, Yong Chen, Lili Mei, An Zhang, Guangting Chen, Guochuan Zhang}



\begin{document}

\maketitle

\begin{abstract}
    In the bottleneck multiple knapsack problem, we are given a set of items and a set of knapsacks, where each item has a profit and a weight, and each knapsack has a capacity. Our goal is to assign items to knapsacks so as to maximize the minimum profit received by any knapsack subject to the capacity constraint.
    When all knapsacks have identical capacity, we give a $(\frac{2}{3} - \varepsilon)$-approximation algorithm for any constant $\varepsilon > 0$. This result almost matches the  $(\frac{2}{3} + \varepsilon)$ inapproximability bound for the bottleneck multiple subset sum problem (Caprara et al., 2000). 
    When the knapsacks can have arbitrary capacities, we propose a $(\frac{1}{2} - \varepsilon)$-approximation algorithm for any constant $\varepsilon > 0$. We also prove a hardness bound of $(\frac{1}{2} + \varepsilon)$ for any constant $\varepsilon > 0$.

\end{abstract}

\section{Introduction}
Maximizing the minimum profit received by multiple agents is a central objective in the study of fairness and equitable resource allocation.
In this paper, we study this objective in the context of knapsack-type constraints.

We consider the \emph{bottleneck multiple knapsack problem} (BMKP), which can be viewed as a max-min variant of the classical multiple knapsack problem (MKP).
Given items with profits and weights and multiple knapsacks with capacities, the classical \emph{multiple knapsack problem} (MKP) maximizes the sum of profits among all knapsacks.
In BMKP, the objective changes to maximizing the minimum total profit received by any knapsack.
Our results address two settings: identical capacities and arbitrary capacities.

\smallskip
\noindent \textbf{Bottleneck multiple subset sum.}
The bottleneck (max--min) objective has been studied in a special case of MKP, namely the \emph{multiple subset sum problem} (MSSP), where each item has profit equal to its weight and all knapsacks have identical capacities.
Caprara, Kellerer, and Pferschy~\cite{caprara2000multiple} presented a polynomial-time approximation scheme (PTAS) for the max-sum objective in MSSP and gave a $\frac{2}{3}$-approximation for bottleneck MSSP.
They also showed that, for any $\varepsilon>0$, achieving a $(\frac{2}{3}+\varepsilon)$-approximation is impossible unless $\mathrm{P}=\mathrm{NP}$.
For the max-sum objective, PTASs for MKP have been studied extensively with successive improvements in running time
(see~\cite{caprara2000multiple, caprara2000ptasDifKnap, caprara20033, kellerer1999polynomial, chekuri2005polynomial, jansen2010parameterized, jansen2012fast}).
In contrast, the bottleneck objective for MKP remains far less understood.
It remains open whether the same $\frac{2}{3}$ ratio can be attained for the more general bottleneck multiple knapsack problem, or whether stronger inapproximability results hold.

\smallskip
\noindent \textbf{Maximin share.}
The bottleneck (max-min) objective also arises naturally in the literature on fair allocation of indivisible goods.
A central notion in this area is the \emph{maximin share} (MMS), which measures the largest value an agent can guarantee for herself by partitioning the items into feasible bundles and then receiving the least valuable bundle.
The MMS criterion has been extensively studied under a variety of feasibility constraints, including cardinality~\cite{biswas2018fair,hummel2022maximin}, matroid~\cite{gourves2017approximate,gourves2019maximin}, and knapsack-type constraints~\cite{deng2024budgeted,hummel2025maximin}, as a fundamental benchmark for fairness.

Recently, Hummel~\cite{hummel2025maximin} studied MMS guarantees in hereditary set system valuations, a generalization of the identical-capacity knapsack (budget) setting.
His result implies a $\frac{2}{5}$-approximation algorithm for identical-capacity BMKP, which is the best previously known guarantee and still far from the $\frac{2}{3}$ hardness barrier.
And it is unclear what approximation ratio can be achieved for arbitrary-capacity BMKP.

\smallskip
\noindent \textbf{Why techniques for the max-sum objective do not apply.}
For max-sum objectives, one may lose an $\varepsilon$ fraction of the total profit by sacrificing the total profits of a few knapsacks.
For instance, under the identical setting, one may apply the APTAS for bin packing~\cite{fernandez1981bin} to pack guessed items into at most $(1+\varepsilon)m+1$ bins where $m$ is the number of knapsacks, 
and then discard the extra bins at an $\varepsilon$-fraction loss in the overall objective~\cite{caprara2000ptasDifKnap}.
Such a strategy fails for the bottleneck objective,
as the loss must be controlled for every knapsack, which calls for finer per-knapsack structure.

\subsection{Our Contributions.}
We give nearly tight approximation guarantees for BMKP in both the identical-capacity and arbitrary-capacity settings.
In the identical-capacity setting, we obtain a $(\frac{2}{3}-\varepsilon)$-approximation algorithm, which nearly matches the $(\frac{2}{3}+\varepsilon)$ inapproximability under the assumption that $\mathrm{P} \neq \mathrm{NP}$~\cite{caprara2000multiple}. Let $|I|$ denote the input size.

\begin{theorem}\label{thm:identical}
    In the identical-capacity setting, for any constant $\varepsilon > 0$, there is a polynomial-time $(\frac{2}{3} - \varepsilon)$-approximation algorithm for the bottleneck multiple knapsack problem with running time $|I|^{2^{O(1/\varepsilon\log (1/\varepsilon))}}$.
\end{theorem}

In the arbitrary-capacity setting, we establish a nearly tight approximation guarantee at $\frac{1}{2}$.
We give a $(\frac{1}{2}-\varepsilon)$-approximation algorithm for any constant $\varepsilon>0$, and show that no polynomial-time $(\frac{1}{2}+\varepsilon)$-approximation is possible unless $\mathrm{P}=\mathrm{NP}$.

\begin{theorem}\label{thm:general}
    For any constant $\varepsilon > 0$, there is a polynomial-time $(\frac{1}{2} - \varepsilon)$-approximation algorithm for the bottleneck multiple knapsack problem with running time $2^{2^{poly(1/\varepsilon)}}\cdot poly(|I|)$.
\end{theorem}

\begin{theorem}\label{thm:hardness}
    Unless $\mathrm{P} = \mathrm{NP}$, the bottleneck multiple knapsack problem has no polynomial-time $(\frac{1}{2} + \varepsilon)$-approximation algorithm for any constant $\varepsilon > 0$.
\end{theorem}

The following table summarizes the results of approximation algorithms for multiple knapsack problems under different objectives and capacity settings.

\begin{table}[!ht]
\begin{tabular}{|c|cc|cc|}
\hline
\multirow{2}{*}{} & \multicolumn{2}{c|}{MKP (max-sum objective)}    & \multicolumn{2}{c|}{BMKP (max-min objective)}    \\ \cline{2-5} 
                  & \multicolumn{1}{c|}{Negative} & Positive & \multicolumn{1}{c|}{Negative} & Positive \\ \hline
Identical capacity  
& \multicolumn{1}{c|}{NP-hard} 
& $1-\varepsilon$~\cite{chekuri2005polynomial, jansen2010parameterized, jansen2012fast} 
& \multicolumn{1}{c|}{$\frac{2}{3} + \varepsilon$~\cite{caprara2000multiple}} 
& $\frac{2}{3} - \varepsilon$ (Theorem~\ref{thm:identical}) \\ \hline
Arbitrary capacity  
& \multicolumn{1}{c|}{NP-hard} 
& $1-\varepsilon$~\cite{chekuri2005polynomial, jansen2010parameterized, jansen2012fast} 
& \multicolumn{1}{c|}{$\frac{1}{2} + \varepsilon$ (Theorem~\ref{thm:hardness})}  
& $\frac{1}{2} - \varepsilon$ (Theorem~\ref{thm:general})\\ \hline
\end{tabular}
\end{table}

\subsection{Related Work.}
\noindent \textbf{Multiple subset sum and multiple knapsack.}
We briefly review approximation schemes for the multiple subset sum problem (MSSP) and the multiple knapsack problem (MKP) with the objective of maximizing the total profit.
Caprara et al.~\cite{caprara2000multiple} gave the first PTAS for MSSP, and also showed that MSSP admits no FPTAS even for two knapsacks unless $\mathrm{P}=\mathrm{NP}$.
Subsequent work extended the PTAS to MKP under the arbitrary capacity setting as well as reduce its running time (see~\cite{caprara2000multiple, caprara2000ptasDifKnap, caprara20033, kellerer1999polynomial, chekuri2005polynomial, jansen2010parameterized, jansen2012fast}). So far, the best-known result is due to Jansen~\cite{jansen2012fast}, which is a PTAS for MKP where knapsacks can have arbitrary capacities, and has a running time of $poly(|I|)+2^{O(1/\varepsilon (\log(1/\varepsilon))^4)}$. This PTAS is near-optimal in the sense that there is no $poly(|I|)+2^{o(1/\varepsilon)}$ time PTAS assuming ETH (Exponential Time Hypothesis).

\smallskip
\noindent \textbf{Machine scheduling.}
Machine scheduling problems can be viewed as a lower-dimensional analogue of multiple knapsack, where one primarily focuses on the profit dimension (i.e., the processing time of jobs).
For identical parallel machines, Hochbaum and Shmoys~\cite{hochbaum1987using} gave a PTAS for the objective of minimizing the maximum completion time (makespan).
Woeginger~\cite{woeginger1997polynomial} presented a PTAS for maximizing the minimum completion time.
For related machines, PTASs are known for both the max-min and min-max objectives~\cite{azar1998approximation,epstein2004approximation}.
For unrelated machines, Lenstra, Shmoys and Tardos~\cite{lenstra1990approximation} gave a $2$-approximation algorithm for makespan minimization, and also proved that no polynomial-time algorithm can achieve an approximation ratio better than $3/2$ unless $\mathrm{P}=\mathrm{NP}$.

\smallskip
\noindent \textbf{Santa Claus problem.}
The max--min variant of scheduling on unrelated machines is commonly known as the \emph{Santa Claus} problem.
Interpreting machines as agents and jobs as items with agent-dependent values, the goal is to allocate items so as to maximize the minimum total value received by any agent.  The best approximation algorithm known for this problem is due to Chakrabarty et al.~\cite{chakrabarty2009allocating} achieving an approximation ratio of $\tilde{\Omega}(\frac{1}{n^{\varepsilon}})$, where $n$ is the number of items. 
And the best known hardness of approximation stands at $\frac{1}{2}$~\cite{bezakova2005allocating,chakrabarty2009allocating}. For more results of the Santa Claus problem, please refer to e.g.~\cite{bansal2006santa, saha2018new, ko2024polynomial}.

\section{Technical Overview.} 
\textbf{Our framework.} Our algorithms follow the standard framework of the PTASs for MKP~\cite{chekuri2005polynomial, jansen2010parameterized, jansen2012fast}.
Assume that the optimal objective value $\mathrm{OPT}$ is known.
 We classify items into two categories:
\begin{itemize}
    \item \emph{Expensive items}: items whose profit is large (say, at least $\varepsilon \cdot \mathrm{OPT}$);
    \item \emph{Cheap items}: items whose profit is small (say, less than $\varepsilon \cdot \mathrm{OPT}$).
\end{itemize}
By this classification, we may assume that each knapsack contains at most $1/\varepsilon$ expensive items.
If some knapsack contains more than $1/\varepsilon$ expensive items, then we can discard the extra ones and still maintain optimality.

Since there are constant number of expensive items per knapsack, we can apply the standard \emph{round-and-guess} approach to them.
We round profits and weights so that each expensive item belongs to one of constantly many types.
We then show that there exists a feasible solution in which every expensive item is replaced by its (rounded) type.
This allows us to guess, for each knapsack, how many items of each type it receives (i.e., a \emph{configuration}), and to pack expensive items according to the guessed configuration while losing only a small amount of profit.

To handle cheap items, we formulate a linear programming.
Since each cheap item has profit below $\varepsilon\cdot \mathrm{OPT}$, losing only $O(1)$ cheap items per knapsack decreases the objective by at most $O(\varepsilon)\cdot \mathrm{OPT}$.
Thus, it suffices to show that we can round an LP solution to an integral assignment while discarding only a constant number of cheap items per knapsack.

For our max--min objective, the main technical challenge can be phrased as follows:
\begin{quotation}\itshape
Is it possible to round the items in such a way that there exists a solution in which, for every knapsack, feasibility is preserved and the loss in profit is small?
\end{quotation}

We discuss the challenges and techniques for the identical capacity setting
and the arbitrary capacity setting separately.

\subsection{Identical Capacities}\label{subsec:identical-overview}

Without loss of generality, we assume the knapsack capacity $B=1$ and $\mathrm{OPT}=1$.
Our goal is to find an assignment such that each knapsack receives total profit at least
$\bigl(\frac{2}{3}-\varepsilon\bigr)$.
Let $m$ be the number of knapsacks.

\smallskip
\noindent\textbf{Challenges.}
Consider an optimal solution.
The main difficulty is that weight rounding (in particular, rounding up) may force us to drop some items to restore feasibility, and the resulting profit loss in a single knapsack may exceed $(1/3+\varepsilon)$.
This can happen for items that are highly sensitive to weight rounding, including
\begin{itemize}
    \item \emph{Heavy items:} items with weight close to $1$ (e.g., $1-o(\varepsilon)$) and profit greater than $(1/3+\varepsilon)$;
    \item \emph{Exact-fit pairs:} two items whose total weight is close to $1$ and whose total profit is close to $1$,
    where each item has profit larger than $(1/3+\varepsilon)$.
\end{itemize}
Assume there are knapsacks containing such items in the optimal solution.
If the weight of any item in the above classes is rounded up, one may have to discard at least one such item, which would incur a profit loss exceeding $(1/3+\varepsilon)$ and violate the desired guarantee.

\smallskip
\noindent\textbf{Our algorithm.} We classify items into heavy/light and expensive/cheap categories.
\begin{itemize}
  \item \textbf{Step 1 (Heavy items).}
  We round heavy items into a constant number of types and guess their assignment to knapsacks.

  \item \textbf{Step 2 (Light-expensive items).}
  Once we know that knapsack $j$ contains a heavy item $h$ with rounded weight $\widetilde w(h)$, we can round the weights of light items using the scaling factor $1-\widetilde w(h)$, and then guess the assignment of light items in each knapsack.

  \item \textbf{Step 3 (Exact-fit pairs).}
  We determine the placement of exact-fit pairs by solving a maximum matching problem.

  \item \textbf{Step 4 (Cheap items).}
  We pack the remaining cheap items by solving a linear program and rounding its fractional solution.
\end{itemize}

\smallskip
\noindent\textbf{Main techniques.}
The main technical component is to show that there exists a solution (of the rounded instance) in which weight rounding is feasible and incurs a profit loss of at most $(1/3+\varepsilon)$ per knapsack.
Starting from an optimal solution, we divide knapsacks into two types depending on whether they contain a heavy item.

\begin{enumerate}[label=(\alph*), leftmargin=*, itemsep=0.3em]
    \item \emph{Knapsacks with only light items:}
    We show that removing items of total profit at most $(1/3+\varepsilon)$ is sufficient to make the weight rounding of the light items feasible.

    \item \emph{Knapsacks with both heavy and light items:}
    \begin{enumerate}[label=(\alph{enumi}\arabic*), ref=\alph{enumi}\arabic*, leftmargin=*, itemsep=0.2em]
        \item \emph{Heavy items:}
        We apply a \emph{rounding} and \emph{shifting} process as follows.

        \begin{itemize}
            \item \emph{Rounding.} 
            We adopt the \emph{linear grouping scheme}
            from bin packing~\cite{fernandez1981bin,williamson2011design}.
            Group the heavy items as follows: the first group contains the heaviest $k$ items, the second group contains the next heaviest $k$ items, and so on.
            Here $k$ is chosen so that the number of groups is constant.
            We discard the first group (to be compensated later), and round the weight of heavy items in each group to the heaviest item in that group.

            \item \emph{Shifting.}
            To maintain feasibility, we shift light items groupwise from their groups to lighter ones (where heavy items are lighter).
        \end{itemize}
        This process guarantees that the weight rounding of heavy items is feasible and the total profit loss is bounded. See Figure~\ref{fig:identical-shifting} for an illustration.
        \item \emph{Light items:}
        We show that removing light items of total profit at most $(1/3+\varepsilon)$ from each knapsack is sufficient to make the weight rounding feasible.
        Moreover, we prove that these removed items can be reassembled and packed into the knapsacks in the discarded group created by the heavy-item shifting step.
  \end{enumerate}
\end{enumerate}

\begin{figure}[tbp]
    \centering
    \includegraphics[width=0.75\textwidth]{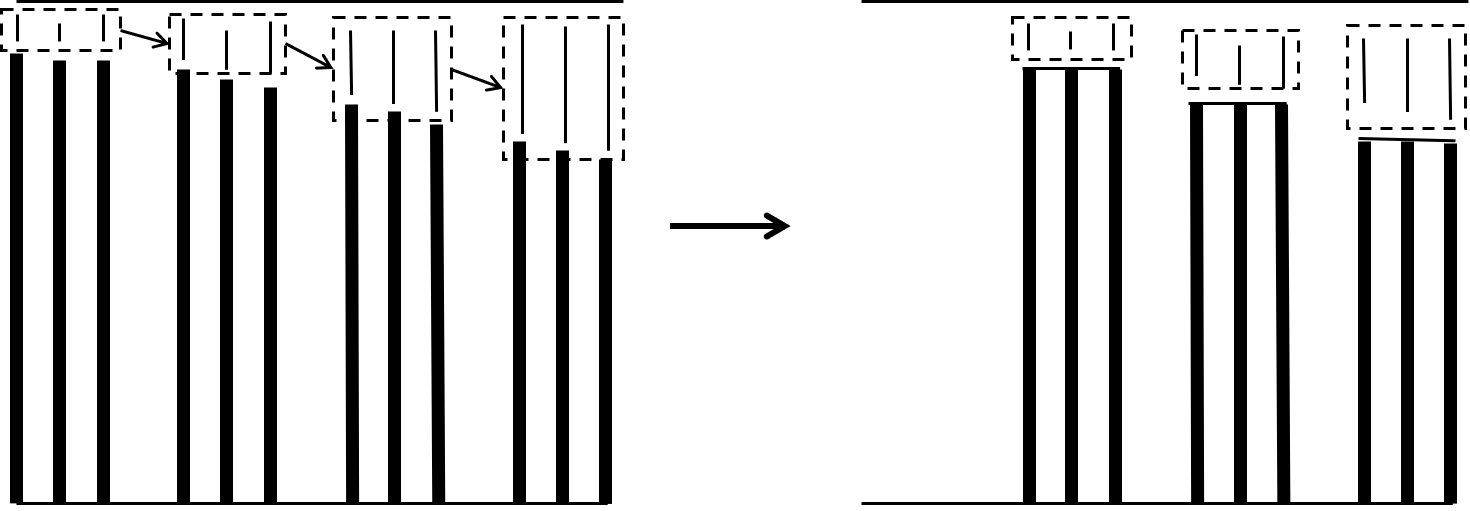}
    \caption{Illustration of the rounding and shifting procedure. Thick (resp., thin) solid segments denote heavy (resp., light) items. Dashed boxes highlight the items being shifted, and arrows show their destinations.}
    \label{fig:identical-shifting}
\end{figure}

\subsection{Arbitrary Capacities}

Without loss of generality, we assume $\mathrm{OPT}=1$ and sort knapsack capacities as
$1=B_1 \le B_2 \le \cdots \le B_m$.
Our goal is to find an assignment such that each knapsack receives total profit at least
$\bigl(\frac{1}{2}-\varepsilon\bigr)$.

\smallskip
\noindent\textbf{Challenges.}
Compared with the identical-capacity case, the main difficulty is that capacities span many scales,
so the standard ``guess configurations + LP'' framework no longer yields a polynomial-time algorithm.

\smallskip
\noindent
\emph{Too many configuration types across capacities.}
Assume that for any fixed capacity value, the number of possible configurations is bounded by a constant, say $\eta $. When capacities vary, different capacity values require different configuration sets.
As a result, a naive enumeration that guesses the configurations for $m$ knapsacks would require up to $\eta^{m}$ possibilities. 

\smallskip
\noindent\textbf{Our algorithm.} Instead of enumerating configurations, we solve a configuration linear program (configuration LP).
\begin{itemize}
  \item \textbf{Step 1 (Rounding).}
  We round item weights/profits and knapsack capacities so that each rounded capacity value admits only constantly many configuration types.
    Overall, the number of distinct configuration types is polynomial in the input size.

  \item \textbf{Step 2 (Configuration LP).}
   We formulate and solve a configuration LP that, for each rounded capacity value, specifies how many knapsacks use each configuration type (possibly fractionally).
   We then convert its solution into a fractional assignment of rounded items (an assignment LP).

  \item \textbf{Step 3 (Rounding the fractional solution).}
    We round the fractional assignment to an integral solution, losing at most $1/2+\varepsilon$ in profit.
\end{itemize}



\smallskip
\noindent\textbf{Main techniques.}
Our main technical step is to round a fractional assignment to an integral solution while preserving feasibility and bounding the profit loss.
If we temporarily ignore the capacity constraints, then the problem becomes a parallel-machine scheduling problem.
Lenstra, Shmoys and Tardos~\cite{lenstra1990approximation} showed that their LP admits a basic feasible solution in which each machine has at most one fractional item.
Assuming every item has profit at most $1/2-\varepsilon$ (otherwise handle it separately), dropping the fractional item per machine loses only a bounded amount.
Our goal is to reduce the original problem to this ``capacity-free'' situation. See Figure~\ref{fig:arb-shifting} for an illustration.

\smallskip
\noindent
\emph{Shifting.}
We partition knapsacks into groups by their rounded capacities, say $\mathcal K_{1},\ldots,\mathcal K_{q}$ in increasing order.
Fix a parameter $\tau$ such that the rounded capacity of $\mathcal K_{j+\tau}$ is sufficiently larger than that of $\mathcal K_{j}$.
We move fractional items from $\mathcal K_{j}$ to $\mathcal K_{j+\tau}$ for all $j$.
As a result, some knapsacks in the first $\tau$ groups may become empty.

\smallskip
\noindent
\emph{Compensation for empty knapsacks.}
The key is that the total number of fractional knapsacks is bounded by a constant.
To achieve this, we ensure that the parameters $\tau$ and $\eta$ are constants, where $\eta$ denotes the number of configuration types in each group.
We then choose a constant-size group $\mathcal K_{0}$ and enforce its configurations to be integral, from which we can remove a small amount of profit to compensate those empty knapsacks.

\begin{figure}[tbp]
    \centering
    \includegraphics[width=0.8\textwidth]{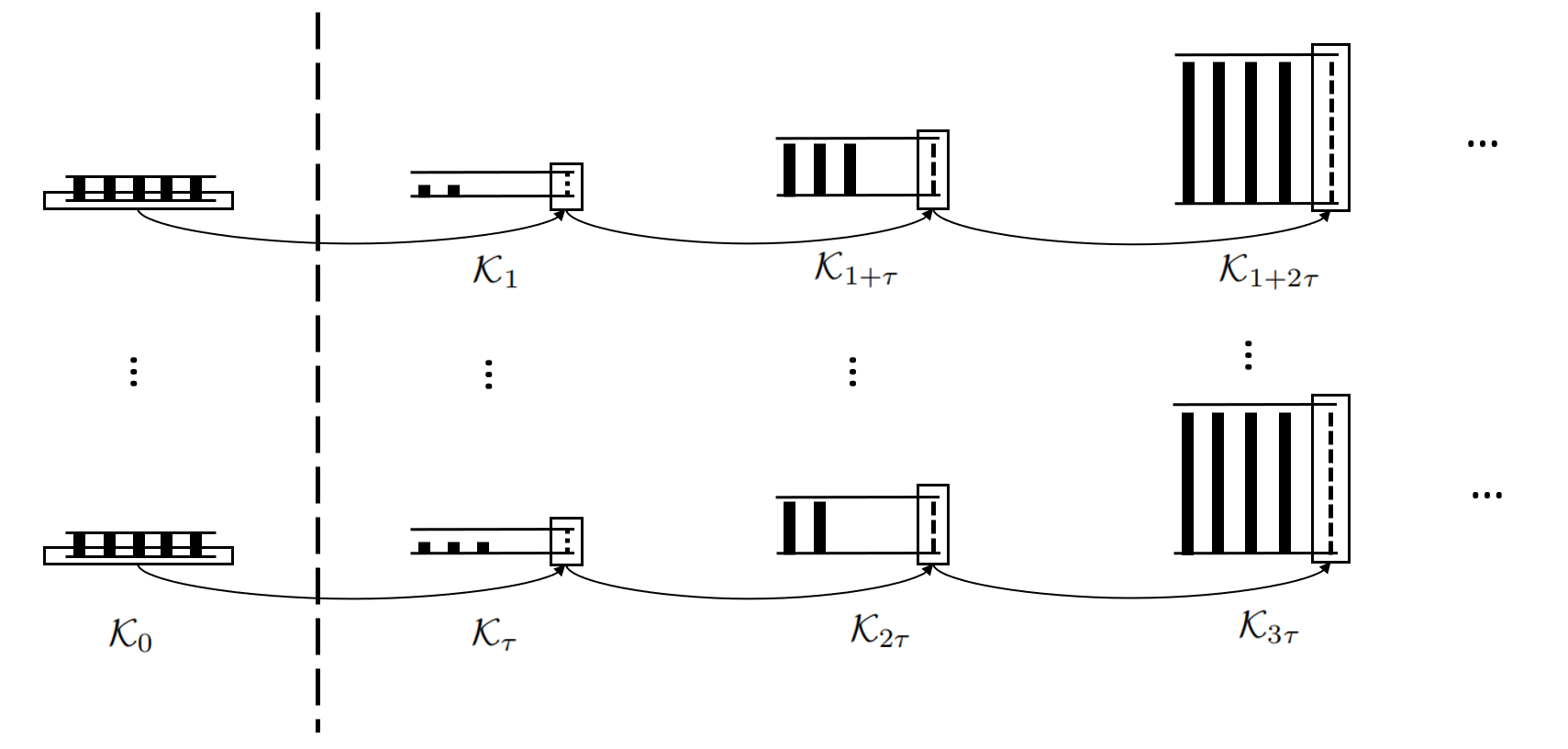}
    \caption{Shifting across capacity groups. Solid (resp., dashed) segments indicate integral (resp., fractional) configurations.
    Arrows show the direction of shifting: fractional items in group $\mathcal K_j$ are moved to group $\mathcal K_{j+\tau}$.
    The group $\mathcal K_0$ consists of a constant number of knapsacks, whose configurations are enforced to be integral and used to compensate $\mathcal{K}_1,\dots,\mathcal{K}_{\tau}$.}
    \label{fig:arb-shifting}
\end{figure}

\section{Notation}
In the bottleneck multiple knapsack problem, we are given a set of $m$ knapsacks with capacity $B_1, \ldots, B_m$, and a set $I$ of $n$ items, where each item $i \in I$ has weight $w(i)$ and profit $p(i)$. For any positive integer $k$, we write $[k]= \{1,2,\ldots,k\}$.
A feasible solution consists of $m$ disjoint subsets $I_1, \ldots, I_m \subseteq I$ with $\sum_{i \in I_j}w(i) \leq B_j$ for each $j \in [m]$. Our goal is to find a feasible solution $(I_1, \ldots, I_m)$ that maximizes the minimum profit received by any knapsack. More precisely, we want to maximize
\(
    \min_{j \in [m]} \sum_{i \in I_j}p(i).
\)

For any subset $I' \subseteq I$ of items, we use $p(I')$ and $w(I')$ to denote the total profit and total weight of the items in $I'$, respectively.

We define the density of an item $i$ to be $\rho(i)=\frac{p(i)}{w(i)}$. If $w(i)=0$, we interpret $i$ as a dummy (nonexistent) item, and set $p(i)=0$ and $\rho(i)=0$.
For any subset $I' \subseteq I$ of items, define its density as $\rho(I')=p(I')/w(I')$.

By scaling, we assume that $\mathrm{OPT}=1$.
In the identical-capacity setting, we normalize the capacity and write $B=1$.
In the arbitrary-capacity setting, we assume $B_1=1$ and $1=B_1 \le B_2 \le \cdots \le B_m$.

\section{An Algorithm for Identical Capacities}
In this section, we prove Theorem~\ref{thm:identical}.

\subsection{Proof Roadmap}
\noindent\textbf{Key definitions.} We introduce the key concepts of our proof and explain their roles at a high level.
\begin{enumerate}[leftmargin=*, itemsep=2pt, topsep=2pt]
\item\textbf{Item categories.}
We will define heavy/light items and expensive/cheap items.
Our main focus is on \emph{heavy-expensive} items and \emph{light-expensive} items.

\item \textbf{Critical items.}
Given a solution, an item is \emph{critical} (with respect to this solution) if it is either expensive or heavy, and packed in a knapsack containing at least three items.
Critical items are exactly the items that we will round.

\item \textbf{Slack solution.}
A solution is said to be \emph{slack} if, in every knapsack, the unused capacity is large enough to allow the rounding of its critical items. 
We will define the precise meaning of ``large enough''.

\item \textbf{Profiles.}
A \emph{profile} (of a slack solution) is a subset of heavy items that serves as the rounded types of critical heavy items.
More precisely, it is defined so that, after replacing each critical heavy item by one of the items in the profile (its rounded type), the resulting solution remains slack.
\end{enumerate}

\smallskip
\noindent\textbf{Organization of the proof.}
We guess the critical heavy items, critical light items, non-critical items successively.
Our argument follows an "existence-guessing" pattern:
before each guessing step, we first prove that there exists a slack solution (containing previously guessed structure); we then guess the assignment of items in this slack solution.
The main technical component is the first existence result.

\begin{enumerate}[leftmargin=*, itemsep=2pt, topsep=2pt]
    \item \textbf{Existence of a slack solution and its profile.}
    We show that there exists a slack solution whose minimum knapsack profit is at least
    $\frac{2}{3}-\varepsilon$, and whose profile has constant cardinality (Lemma~\ref{lem:T-exist}).
    Since the profile has constant size, we can enumerate all possibilities of the profile in polynomial time (Lemma~\ref{lem:T-guess}).

    \item \textbf{Guessing critical heavy items.}
    There exists a slack solution with objective value at least $\frac23-2\varepsilon$ whose profile is the guessed one (Lemma~\ref{lem:easy-hc-exist}).
    Moreover, for each rounded type in the profile, every critical heavy item in this solution is heavier than any non-critical heavy item of the same type (Lemma~\ref{lem:easy-hc-exist}). 
    We show that these critical heavy items can be guessed (Lemma~\ref{lem:hc-guess}).

    \item \textbf{Guessing critical light items.}
    There exists a slack solution with objective value at least $\frac23-3\varepsilon$ whose critical heavy items are the guessed ones (Lemma~\ref{lem:easy-cl-exist}).
    Moreover, for each rounded type of light items, every critical light item in this solution is heavier than any non-critical light item of the same type (Lemma~\ref{lem:easy-cl-exist}). 
    We show that these critical light items can be guessed (Lemma~\ref{lem:cl-guess}).

    \item \textbf{Determining non-critical items.}
    Finally, we complete the packing by handling knapsacks with $|I_j|=2$ (Lemma~\ref{lem:pair-by-matching-short}) and assigning the  non-critical cheap items
    in knapsacks with $|I_j|\ge 3$ (Lemma~\ref{lem:sliding}).
\end{enumerate}

\subsection{Preliminaries}
Recall that $\mathrm{OPT}=1$ and $B=1$.
Our goal is to compute, in polynomial time, a feasible solution $(I_1,\ldots,I_m)$ such that
$p(I_j)\ge \frac{2}{3}-6\varepsilon$ for all $j\in[m]$.
(The constant factor in front of $\varepsilon$ is for convenience; by rescaling $\varepsilon$, one can achieve $\frac23-\varepsilon$.)
For simplicity, we assume that $\varepsilon$ is a sufficiently small constant and $1/\varepsilon$ is an integer.

\smallskip
\noindent\textbf{Assumption on item profits.}
We assume that every item $i$ satisfies $p(i) < \frac{2}{3}-6\varepsilon$.
Otherwise, if some item has profit at least $\frac{2}{3}-6\varepsilon$, we can assign it to an empty knapsack and remove
both the item and that knapsack from the instance.
The remaining instance still admits a solution with objective value at least $1$. 
Consequently, we can assume that every knapsack in the optimal solution contains at least two items.

\begin{definition}[Heavy, Light, Expensive, and Cheap Items]\label{def:heavy-light-expensive-cheap}
An item $i$ is heavy if $w(i)\ge 1-\varepsilon^2$, and light otherwise.
It is expensive if $p(i)\ge \varepsilon$, and cheap otherwise.
\end{definition}

\noindent\textbf{Assumption on number of expensive items in each knapsack.}
We assume that each knapsack contains at most $1/\varepsilon$ expensive items. If a knapsack contains more, we keep an arbitrary subset of $1/\varepsilon$ expensive items, whose total profit is at least $1$.

\begin{definition}[Critical Items] \label{Def:critical-items}
    Let $\mathcal{S} = (I_1, \ldots, I_m)$ be a feasible solution. 
    The critical items of $\mathcal{S}$ are defined to be the heavy items and the expensive items in any $I_j$ with $|I_j| \geq 3$. 
    All the items in $\mathcal{S}$ that are not critical are called non-critical items. 
\end{definition}

\noindent\textbf{Item categories.}
There are four types of items: \emph{heavy-expensive}, \emph{heavy-cheap}, \emph{light-expensive}, and \emph{light-cheap}.
We assume w.l.o.g.\ that heavy-cheap items do not exist.
By Definition~\ref{def:heavy-light-expensive-cheap}, each knapsack contains at most one heavy item. If this heavy item is cheap,
removing it decreases the profit of that knapsack by less than $\varepsilon$; by rescaling $\varepsilon$, we can still state the guarantee as at least $\frac{2}{3}-\varepsilon$.
Therefore, every heavy item is expensive.
Combined with Definition~\ref{Def:critical-items}, it follows that every critical item is expensive, and thus each knapsack contains at most $1/\varepsilon$ critical items.

\begin{definition}[Slack Solutions]\label{Def:SlackSolution}
    We say that a solution $(I_1, \ldots, I_m)$ is slack if for any $j \in [m]$ with $|I_j| \geq 3$, exactly one of the following conditions holds.
    \begin{enumerate}[label={\normalfont (\roman*)}]
        \item  All the items in $I_j$ are light, and $w(I_j) \leq 1 - \varepsilon^3$, 

        \item $I_j$ contains a heavy item $h$, and
        \[
           w(I_j) - w(h) \leq (1 -  \varepsilon)(1 - w(h)),
           \quad\text{or equivalently,}\quad
           1-w(I_j)\ge \varepsilon (1-w(h)).
        \]
    \end{enumerate}
\end{definition}

The first case in Definition~\ref{Def:SlackSolution} implies that the unused capacity $1-w(I_j)$ is at least $\varepsilon^3$. The second case implies that the total weight of light items $w(I_j)-w(h)$ occupy at most a $(1-\varepsilon)$-fraction of $(1-w(h))$. In other words, the unused capacity $1-w(I_j)$ is at least $\varepsilon(1-w(h))$.

\begin{definition}[Profile]\label{def:profile}
    Let $T \subseteq I$ be a subset of heavy items. Let $\mathcal{S} = (I_1, \ldots, I_m)$ be a slack solution. We say that $T$ is a profile of $\mathcal{S}$ if, for any $I_j$ that contains a critical heavy item $h$, there exists $t \in T$ such that $\lfloor p(t)/\varepsilon \rfloor  = \lfloor p(h)/\varepsilon \rfloor$ and 
        \begin{align}\label{eq:profile}
            w(I_j) - w(h) \leq (1 - \varepsilon)(1 - w(t)) \leq (1 - \varepsilon)(1 - w(h)).
        \end{align}
\end{definition}

Note that the inequality above guarantees that $w(t)\ge w(h)$. It implies that if we replace the heavy item $h$ by $t$, we have 
\[
w\bigl((I_j\setminus\{h\})\cup\{t\}\bigr)-w(t)
= w(I_j)-w(h)
\le (1-\varepsilon)\bigl(1-w(t)\bigr).
\]
Therefore, if we round $h$ to $t$, the resulting solution is still slack.


\subsection{Existence of a Structured Slack Solution}
We show that there is a slack solution that has a profile of constant size.

\begin{lemma}\label{lem:T-exist}
    There is a slack solution $\mathcal{S} = (I_1, \ldots, I_m)$ such that $p(I_j) \geq \frac{2}{3} - \varepsilon$ for any $j \in [m]$ and $\mathcal{S}$ has a profile $T \subseteq I$ of heavy items with $|T| \leq O((\frac{1}{\varepsilon})^2)$.
\end{lemma}
\begin{proof}
    We first show the existence of a slack solution $\mathcal{S}$ with $p(I_j) \geq \frac{2}{3} - \varepsilon$ for any $j \in [m]$.
    Then we show the existence of the profile.

    \noindent\textbf{Step 1: Existence of a slack solution.}
    Given an optimal solution $(X_1^*,\dots,X_m^*)$ with $p(X_{j}^*)\ge 1$ for all $j\in[m]$, we convert it into a slack soluton with $p(I_j) \geq \frac{2}{3} - \varepsilon$ for any $j \in [m]$. Without loss of generality, we assume every $X_j^*$ is minimal in the sense that deleting any item makes its total profit less than $1$.

    By assumption, each item has profit of at most $\frac{2}{3}-6\varepsilon$ and $|X_j^*|\ge 2$ for any $j\in [m]$. If there exists $j\in[m]$ such that $X_j^*$ contains two items $i_1,i_2$ with $p(i_1)+p(i_2)\ge \frac{2}{3}-\varepsilon$, then we simply set $I_j=\{i_1,i_2\}$. 
    By Definition~\ref{Def:critical-items}, they are non-critical items and does not affect slackness or the profile.
    
    We thus restrict attention to the case where, for all $i_1,i_2\in X_j^*$,
    \begin{align}\label{cond:1}
    p(i_1)+p(i_2)< \frac{2}{3}-\varepsilon.
    \end{align}
    We distinguish between the two cases below.
      
    \smallskip
    \textbf{(Case 1) All items in $X_j^*$ are light.}
    If $w(X_j^*)\le 1-\varepsilon^3$, then let $I_j=X_j^*$; by Definition~\ref{Def:SlackSolution}, this knapsack is slack, and we are done.
    Therefore, we assume $w(X_j^*)> 1-\varepsilon^3$ and partition the items in $X_j^*$ into three sets:
    \[
    X_e=\{\,i\in X_j^* \mid p(i)>\tfrac{1}{3}+\varepsilon \text{ and } w(i)\le 1-\varepsilon^2\,\},
    \]
    where $w(i)\le 1-\varepsilon^2$ holds since all items are light,
    \[
    X_c=\{\,i\in X_j^* \mid p(i)\le \tfrac{1}{3}+\varepsilon \text{ and } w(i)\le \varepsilon^3\,\},
    \]
    and
    \[
    X_d=\{\,i\in X_j^* \mid p(i)\le \tfrac{1}{3}+\varepsilon \text{ and } w(i)>\varepsilon^3\,\}.
    \]
    By construction, $X_j^* = X_e \cup X_c \cup X_d$.
    If $X_d\neq\emptyset$, then we can pick any item $i\in X_d$ and set $I_j=X_j^*\setminus\{i\}$; by Definition~\ref{Def:SlackSolution}, we are done.
    Therefore, we assume that $X_d=\emptyset$.

   Note that there is at most one item in $X_{e}$, since otherwise there would be two items whose total profit is at least $\frac{2}{3}-\varepsilon$, contradicting equation~\eqref{cond:1}. This implies
    \begin{align}\label{eq:Yexist}
        w(X_{c})= w(X_j^*)-w(X_{e})\ge (1-\varepsilon^3)-(1-\varepsilon^2)=\varepsilon^2-\varepsilon^3\ge \varepsilon^3.
    \end{align}
    Let $Y\subseteq X_c$ be obtained by sorting the items in $X_c$ in increasing order of density and adding them until the total weight first reaches at least $\varepsilon^3$. The set $Y$ always exists due to equation~\eqref{eq:Yexist}, and it follows that 
    \begin{align}\label{eq:Yweight}
        \varepsilon^3\le w(Y)\le 2\varepsilon^3.
    \end{align}
    Since $Y$ consists of items with the smallest density, the density of $Y$ is at most the density of $X_{c}$. Recall that $X_j^*$ is minimal. Fix any $i\in X_c$. Then $p(X_c\setminus\{i\}) \le p(X_j^*\setminus\{i\})<1$, and since $p(i)\le \tfrac13+\varepsilon$, we get $p(X_c) < 1+\tfrac13+\varepsilon$.
    Hence, by equation~\eqref{eq:Yexist} and~\eqref{eq:Yweight}, we have
    \begin{align}
        p(Y)\le w(Y)\cdot \frac{p(X_{c})}{w(X_{c})}\le 2\varepsilon^3\cdot \frac{1+\frac13+\varepsilon}{\varepsilon^2-\varepsilon^3}\le\frac{1}{3},\nonumber
    \end{align}
    where the last inequality holds for $0<\varepsilon\le \frac{-9+\sqrt{105}}{12}$. Recall that we assume $\varepsilon$ is a sufficiently small constant, so the above inequality holds.
    
    Then we let $I_j=X_j^*\setminus Y$. It is clear that $p(I_j)= p(X_j^*)-p(Y)\ge \frac{2}{3}-\varepsilon$ and $w(I_j)= w(X_j^*)-w(Y)\le 1-\varepsilon^3$, which implies that $I_j$ is slack.

    \smallskip
    \textbf{(Case 2) $X_j^*$ contains a heavy item.}
    By Definition~\ref{def:heavy-light-expensive-cheap}, each knapsack contains at most one heavy item. Denote it by $h$.
    If $p(h)\le \frac{1}{3}+\varepsilon$, let $I_j = X_j^*\setminus\{h\}$.
    Then $I_j$ contains only light items, and
    \[
    p(I_j)=p(X_j^*)-p(h)\ge \frac{2}{3}-\varepsilon,
    \qquad
    w(I_j)=w(X_j^*)-w(h)\le 1-\varepsilon^3.
    \]
    By Definition~\ref{Def:SlackSolution}, $I_j$ is slack.
    It remains to consider the case where $p(h)> \frac{1}{3}+\varepsilon$.

    \begin{claim}\label{Clm:RemovedSets}
        If $p(h)> \frac{1}{3}+\varepsilon$, there is a subset $Y\subseteq X_j^*\setminus\{h\}$ with 
        \begin{align}\label{eq:removed-set-heavy}
            \varepsilon\le p(Y)\le \frac{1}{3} \quad\text{ and  }\quad \varepsilon(w(X_j^*)-w(h))\le w(Y)\le \varepsilon^2.
        \end{align}
    \end{claim}
    \begin{proof}[Proof of Claim~\ref{Clm:RemovedSets}]
        Let $X_{light}=X_j^*\setminus \{h\}$. 
        Since $p(h)>\frac{1}{3}+\varepsilon$, every item $i\in X_{\mathrm{light}}$ satisfies $p(i)\le \frac{1}{3}-\varepsilon$. Otherwise $p(h)+p(i)\ge \frac{2}{3}$, contradicting~\eqref{cond:1}.
         Recall that $p(h)\le \frac{2}{3}-\varepsilon<\frac23$ by assumption and $w(h)\ge 1-\varepsilon^2$ by Definition~\ref{def:heavy-light-expensive-cheap}.  
        It follows that
        \begin{align}\label{eq:X-light}
            p(X_{light})=p(X_j^*)-p(h)\ge \frac{1}{3} \text{ and } w(X_{light})=w(X_j^*)-w(h)\le 1-(1-\varepsilon^2)=\varepsilon^2.
        \end{align}

        Let $Y'\subseteq X_{light}$ be obtained by sorting the items in $X_{light}$ in increasing order of density and adding them until $p(Y')> 1/3$. Let the last item added to $Y'$ be $l$, and set $Y=Y'\setminus\{l\}$. Recall that every item in $X_{light}$ has a profit of at most $\frac{1}{3}-\varepsilon$. We have
        \begin{align}\label{eq:Y-at-least-eps}
            \frac{1}{3}\ge p(Y)\ge \frac{1}{3}-p(l)\ge\frac13-(\frac{1}{3}-\varepsilon)=\varepsilon.
        \end{align}

        Since $Y$ consists of items of the smallest density, its density is at most the density of $X_{light}$. 
        Moreover, by minimality of $X_j^*$, $p(X_{\mathrm{light}})=p(X_j^*\setminus\{h\})<1$.
        Therefore,
        \begin{align}
            w(Y)\ge p(Y)\cdot\frac{w(X_{light})}{p(X_{light})}>\varepsilon\cdot w(X_{light})=\varepsilon(w(X_j^*)-w(h)).\nonumber
        \end{align}
        The strict inequality follows from~\eqref{eq:Y-at-least-eps} and $p(X_{\mathrm{light}})<1$.
        Finally, equation \eqref{eq:X-light} implies $w(Y)\le w(X_{\mathrm{light}})\le \varepsilon^2$.
        Combining the above inequalities yields~\eqref{eq:removed-set-heavy}, completing the proof.
    \end{proof}

    We set $I_j=X_j^*\setminus Y$ where $Y$ is defined in Claim~\ref{Clm:RemovedSets}. It is clear that $p(I_j)=p(X_j^*)-p(Y)\ge \frac{2}{3}-\varepsilon$ and 
    \begin{align}
        w(I_j)-w(h)=w(X_j^*)-w(Y)-w(h)\le(1-\varepsilon)(w(X_j^*)-w(h))\le(1-\varepsilon)(1-w(h)),\nonumber
    \end{align}
    which implies that $I_j$ is slack.
    
    \smallskip
    \noindent\textbf{Step 2: Existence of a profile of small cardinality.}
    We have constructed a slack solution $\mathcal{S}=(I_1,\dots,I_m)$ with $p(I_j)\ge \frac{2}{3}-\varepsilon$ for all $j\in[m]$.
    Next, we convert $\mathcal{S}$ into another slack solution $\mathcal{S}'=(I_1',\dots,I_m')$ whose profile has constant cardinality.

    Consider the knapsacks that contain a heavy item in the slack solution $\mathcal{S}$.
    Let $\mathcal{J}_H\subseteq [m]$ be the set of indices $j$ such that $I_j$ contains a heavy item.
    Recall that each such $I_j$ contains exactly one heavy item. We denote it by $h_j$.

    We partition $\mathcal{J}_H$ into groups according to the value of $\lfloor p(h_j)/\varepsilon\rfloor$:
    two indices $j,j'\in\mathcal{J}_H$ are in the same group if
    $\lfloor p(h_j)/\varepsilon\rfloor=\lfloor p(h_{j'})/\varepsilon\rfloor$.
    Note that there are at most $1/\varepsilon$ such groups since every heavy item has profit at most $\frac{2}{3}-6\varepsilon$ and at least $\varepsilon$.
    We handle each group separately as follows.

    Fix an integer $t$ and consider the group
    $\mathcal{K}=\{\,j\in\mathcal{J}_H \mid \lfloor p(h_j)/\varepsilon\rfloor = t\,\}$.
    Write $\mathcal{K}=\{1,2,\dots,k\}$ after relabeling the corresponding knapsacks,
    and further relabel so that
    \begin{align}\label{eq:h-increasing-order}
    w(h_1)\le w(h_2)\le \cdots \le w(h_k).
    \end{align}

    For each $j\in[k]$, let $Y_j\subseteq X_j^*\setminus \{h_j\}$ be the set of light items guaranteed by Claim~\ref{Clm:RemovedSets}.
    Recall that $I_j = X_j^*\setminus Y_j$.
    We now construct new sets $I_1',\dots,I_k'$ as follows.

    \smallskip
    \noindent$\bullet$ For $1\le j\le \lfloor \varepsilon k\rfloor$, let $\ell=\left\lceil\frac{2}{3\varepsilon}\right\rceil$ and define
    \[
    I_j' \;=\; Y_{(j-1)\ell+1}\ \cup\  Y_{(j-1)\ell+2}\ \cup\ \cdots\ \cup\ Y_{j\ell}.
    \]
    By Claim~\ref{Clm:RemovedSets}, each $Y_j$ consists only of light items and satisfies
    $p(Y_j)\ge \varepsilon$ and $w(Y_j)\le \varepsilon^2$.
    Hence,
    \[
    p(I_j') \ge \ell\varepsilon \ge \frac{2}{3}
    \qquad\text{and}\qquad
    w(I_j') \le \ell\varepsilon^2 \le \varepsilon \le 1-\varepsilon^3,
    \]
    where $\ell\varepsilon^2 \le \varepsilon \le 1-\varepsilon^3$ holds for $0<\varepsilon\le \frac13$.
    Therefore, $I_j'$ is slack by Definition~\ref{Def:SlackSolution}.

    \smallskip
    \noindent$\bullet$ For $\lfloor \varepsilon k\rfloor < j \le k$, let $s=j-\lfloor\varepsilon k\rfloor$ and set
    \[
    I_j' \;=\; \bigl(I_j\setminus\{h_j\}\bigr)\ \cup\ \{h_s\}.
    \]
    That is, we replace the heavy item $h_j$ by $h_s$.
    Since $s<j$, by equation~\eqref{eq:h-increasing-order}, we have $w(h_s)\le w(h_j)$, and thus $w(I_j')\le w(I_j)$.
    
    Note that $\lfloor p(h_j)/\varepsilon\rfloor=\lfloor p(h_{s})/\varepsilon\rfloor$, which implies $p(h_s)\ge p(h_j)-\varepsilon$. By Claim \ref{Clm:RemovedSets} and $I_j=X_j^*\setminus Y_j$, the profit of $I_j$ is
    \begin{align}
        p(I_j')=p(I_j)+p(h_s)-p(h_j)\ge p(X_j^*)-p(Y_j)-\varepsilon\ge 1-\frac{1}{3}-\varepsilon=\frac{2}{3}-\varepsilon.\nonumber
    \end{align}
    Meanwhile, we have 
    \begin{align}\label{eq:weight-profile}
        w(I_j')-w(h_s)&=w(I_j)-w(h_j)\nonumber\\
        &\le (1-\varepsilon)(1-w(h_j))\nonumber \\
        &\le (1-\varepsilon)(1-w(h_{s})).
    \end{align}
    The first inequality follows from the slackness of $I_j$ in $\mathcal{S}$.
    The last inequality holds since $w(h_s)\le w(h_j)$.
    By Definition~\ref{Def:SlackSolution}, $I_j'$ is slack.

    Above all, we obtained a slack solution $\mathcal{S'}=(I_1',\dots,I_m')$.
    Next, we construct its profile $T$.
    Let $\delta=\max\{1,\lfloor\varepsilon k\rfloor\}$.
    Let $\{h_{\delta},h_{2\delta},\dots,h_{r\delta}\}$ be the representative heavy items of this group,
    where $r=\left\lfloor \frac{k}{\delta} \right\rfloor=O(\frac{1}{\varepsilon})$.

    For a group with $\lfloor p(h_j)/\varepsilon\rfloor=t$, we set
    $
    T_t=\{h_{\delta},h_{2\delta},\dots,h_{r\delta}\}
    $
    to be the set of representative heavy items defined above.
    Here $1\le t\le  1/\varepsilon$ because by assumption, $\varepsilon\le p(h_j)\le \frac{2}{3}-6\varepsilon$ for any heavy item $h_j$.
    We then define the profile set $T$ to be the union of $T_t$ over all groups, i.e.,
    $
    T \;=\; \bigcup_{t=1}^{ 1/\varepsilon} T_t.
    $
    Since $r=O(\frac{1}{\varepsilon})$ for each group, we have $|T|=O\big((\frac{1}{\varepsilon})^2\big)$.

    It remains to argue that $T$ is the profile of $\mathcal{S'}=(I_1',\dots,I_m')$.
    Towards this, it suffices to show that equation~\eqref{eq:profile} holds for each $T_t$ where $1\le t\le  1/\varepsilon$. To see this, observe that
   for each heavy item $h_{s}\in I_j$ where $s=j-\lfloor\varepsilon k\rfloor$, there must exist some $h_{j'}\in T$ with $s\le j'\le j$, which implies $w(h_{s})\le w(h_{j'})\le w(h_j)$.
    Therefore, by equation~\eqref{eq:weight-profile}, we have
    \begin{align}
        w(I_j')-w(h_s)&=w(I_j)-w(h_j)\nonumber\\
        &\le (1-\varepsilon)(1-w(h_j))\nonumber\\
        &\le (1-\varepsilon)(1-w(h_{j'}))\nonumber\\
        &\le (1-\varepsilon)(1-w(h_{s})).\nonumber
    \end{align}
    Combined with $\lfloor p(h_{j'})/\varepsilon\rfloor=\lfloor p(h_{s})/\varepsilon\rfloor$, we conclude that $T$ is the profile of $\mathcal{S'}$ and $|T|=O\big((\frac{1}{\varepsilon})^2\big)$.
    The constructed solution $\mathcal{S'}$ is exactly the solution $\mathcal{S}$ claimed in the lemma.
\end{proof}

\begin{lemma}\label{lem:T-guess}
    After guessing for $n^{O\big((\frac{1}{\varepsilon})^2\big)}$ rounds, at least one of the guesses gives the $T$ that is the profile of the slack solution in Lemma~\ref{lem:T-exist}. 
\end{lemma}
\begin{proof}
    Since $|T|=O\big((\frac{1}{\varepsilon})^2\big)$, and according to the definition of profile, every element of $T$ must be one of the heavy items (i.e., at most $n$ possibilities), we can guess $T$ via $n^{O\big((\frac{1}{\varepsilon})^2\big)}$ enumerations.
\end{proof}

\subsection{Guessing Critical Heavy Items}\label{subsec:heavy}
We have correctly guessed the profile $T \subseteq I$ satisfying Lemma~\ref{lem:T-exist}. 
Our goal in this subsection is to guess the critical heavy items using $T$. 

In particular, we associate each heavy item $h\in I$ with a unique item $\mathcal{T}(h)\in T$.
We define the \emph{label weight} and \emph{label profit} of $h$ to be the weight and profit of $\mathcal{T}(h)$, respectively.

\begin{definition}[Label weight and label profit]\label{def:label}
Fix a profile $T\subseteq I$.
For a heavy item $h$, let
\[
\mathcal{T}(h)=\left\{\, t\in T \;\middle|\;
\Big\lfloor \tfrac{p(t)}{\varepsilon}\Big\rfloor=\Big\lfloor \tfrac{p(h)}{\varepsilon}\Big\rfloor
\ \text{and}\ w(t)\ge w(h)\right\}.
\]
If $\mathcal{T}(h)\neq\emptyset$, choose any
$t\in\arg\min_{x\in\mathcal{T}(h)} w(x)$ and define
$\widetilde w(h)=w(t)$ and $\widetilde p(h)=p(t)$.
Otherwise, define $\widetilde w(h)=+\infty$ and $\widetilde p(h)=0$.
\end{definition}

\noindent\textbf{Remark.}
One can view item $t$ as the result of rounding $h$ to a nearby item in the profile $T$.
Note that we neither alter the instance nor replace $h$ by $t$ here.
Instead, we simply introduce the notations $\widetilde w(h)$ and $\widetilde p(h)$ to label $h$.

\begin{definition}[Type and $\prec$-order]\label{def:prec-order}
Two heavy items $h_1,h_2$ are of the same \emph{type} if
$\widetilde w(h_1)=\widetilde w(h_2)$ and $\widetilde p(h_1)=\widetilde p(h_2)$.
We define a partial order $\prec$ on heavy items as follows:
for heavy items $h_1,h_2$, we write $h_1\prec h_2$ iff $h_1$ and $h_2$ have the same type and
$w(h_1)<w(h_2)$.
Given a set $A$ of heavy items, an item $h\in A$ is \emph{$\prec$-minimal} (resp., \emph{$\prec$-maximal})
if there is no $h'\in A$ with $h'\prec h$ (resp., $h\prec h'$).
\end{definition}

Note that we cannot directly guess which heavy items are critical within each type because it may take exponential time.
Instead, we may guess the number of critical heavy items of each type and select them arbitrarily.
To guarantee that such a selection is valid, we need to show that even if we select the heaviest items of each type as critical heavy items, the resulting solution is slack and have the desired objective value.

The following lemma shows that there exists a slack solution in which, within each type, every critical heavy item is heavier than every non-critical heavy item (Property~\ref{lem:easy-hc-exist-p3}).
Moreover, the total profit of this solution decreases by at most $\varepsilon$ compared to the solution in Lemma~\ref{lem:T-exist} (Property~\ref{lem:easy-hc-exist-p1}).
And the property of profile is preserved (Property~\ref{lem:easy-hc-exist-p2}). 
Here we replace the notation $w(t)$ in equation~\eqref{eq:profile} by $\widetilde w(h)$ .

\begin{restatable}{lemma}{lemeasyhcexist}\label{lem:easy-hc-exist}
    There exists a slack solution $\mathcal{S} = (I_1, \ldots, I_m)$ satisfying the following.
    \begin{enumerate}[label = {\normalfont (\roman*)}]
        \item $p(I_j) \geq \frac{2}{3} - 2\varepsilon$ for any $j \in [m]$.\label{lem:easy-hc-exist-p1}

        \item For each $j \in [m]$, if $I_j$ contains a critical heavy item $h$, then 
        \[
            w(I_j) - w(h)  \leq (1 - \varepsilon)(1 - \widetilde{w}(h)).
        \] \label{lem:easy-hc-exist-p2}

        \item For any critical heavy item $h$, and for any non-critical heavy item $h'$, if $h$ and $h'$ are of the same type, then $h \succ h'$.\label{lem:easy-hc-exist-p3}
    \end{enumerate}
\end{restatable}
\begin{proof}
(Deferred to Appendix~\ref{app:easy-hc-exist}.)
\end{proof}

\begin{lemma}\label{lem:hc-guess}
    After guessing $O(n^{(\frac{1}{\varepsilon})^2})$ rounds, 
    at least one of the guesses gives $m$ disjoint set $H^*_1, \ldots, H^*_m$ of heavy items such that there exists a slack solution $\mathcal{S} = (I_1, \ldots, I_m)$ satisfying the following.
    \begin{enumerate}[label = {\normalfont (\roman*)}]
        \item $p(I_j) \geq \frac{2}{3} - 2\varepsilon$ for any $j \in [m]$. \label{lem:hc-guess-p1}

        \item For each $j \in [m]$, if $I_j$ contains a critical heavy item $h$, then 
        \[
            w(I_j) - w(h)  \leq (1 - \varepsilon)(1 - \widetilde{w}(h)).
        \] \label{lem:hc-guess-p2}

        \item The set of critical heavy items of $\mathcal{S}$ in $I_j$ is exactly $H^*_j$.\label{lem:hc-guess-p3}
    \end{enumerate}
\end{lemma}
\begin{proof}
    We first guess the number of critical heavy items of each type in the slack solution $\mathcal{S}$ satisfying Lemma~\ref{lem:easy-hc-exist}.
    Note that the number of heavy-item types is at most $|T|=O((1/\varepsilon)^2)$.
    Index these types by $\tau_1,\ldots,\tau_{|T|}$.
    For each $\ell\in\big[|T|\big]$, let $n_\ell$ denote the number of critical heavy items of type $\tau_\ell$ in $\mathcal{S}$, and
    let $(n_1,\ldots,n_{|T|})$ be the count vector.
    We enumerate all possibilities for $(n_1,\ldots,n_{|T|})$. Since $0\le n_\ell\le n$ for every $\ell$, this takes
    at most $n^{|T|}=n^{O((1/\varepsilon)^2)}$ rounds.

    Fix the round that guesses the correct tuple. In this round, for every $\ell\in[|T|]$ we select the heaviest $n_\ell$ items of type $\tau_\ell$. Call the union of all selected items $H^*$.
    By Lemma~\ref{lem:easy-hc-exist}, the set of critical heavy items used by $\mathcal{S}$ is exactly $H^*$.

    We distribute the items of $H^*$ among $m$ sets $H_1^*,\ldots,H_m^*$, allowing empty sets, so that $|H_j^*|\le 1$ for all $j\in[m]$ and $\bigcup_{j=1}^m H_j^* = H^*$.
    We can relabel the knapsacks in $\mathcal{S}$ so that property~\ref{lem:hc-guess-p3} is satisfied. Property~\ref{lem:hc-guess-p1} and~\ref{lem:hc-guess-p2} follow directly from Lemma~\ref{lem:easy-hc-exist}.
\end{proof}

Note that $H^*_j$ either contains one heavy item or is empty, meaning that there is no critical heavy item in knapsack $j$.


\subsection{Guessing Critical Light Items}\label{subsec:light}
Given $H^*_1, \ldots, H^*_m$ satisfying Lemma~\ref{lem:hc-guess}, our next goal is to guess the critical light items. 
Similarly to heavy items, we define label weight, label profit, type and $\prec$-order for light items.

Let $e$ be a light-expensive item.
We define $W(e)$ to be the scaling factor for rounding the weight of $e$.
To keep the solution feasible after rounding, the factor $W(e)$ should depend on the heavy item (if any) packed together with $e$, which is not known in advance.
We therefore use a conservative choice: define $t$ to be the heaviest item in the profile $T$ that $w(e)\le 1-w(t)$, and then let $W(e)$ be $\varepsilon(1-\widetilde w(t))$, where $\widetilde w(t)$ is given in Definition~\ref{def:label}.
If no such $t$ exists, we use a fixed scaling factor $\varepsilon^3$.

\begin{definition}[Label weight and label profit for light items]\label{def:label-light}
Fix a profile $T\subseteq I$.
Let $e$ be a light-expensive item and define $\mathcal{T}(e) = \{\, t\in T \mid w(e)\le 1-w(t)\,\}$.
If $\mathcal{T}(e)\neq\emptyset$, let $t\in\arg\max_{x\in \mathcal{T}(e)} w(x)$ and set
\[
W(e)=\varepsilon\bigl(1-\widetilde w(t)\bigr).
\]
Otherwise, set $W(e)=\varepsilon^3$.
We define the \emph{label weight} and \emph{label profit} of $e$ by
\[
\widetilde w(e)=\left\lceil\frac{w(e)}{\varepsilon W(e)} \right\rceil\cdot \varepsilon W(e),
\qquad
\widetilde p(e)=\left\lfloor\frac{p(e)}{\varepsilon^2}\right\rfloor\cdot \varepsilon^2.
\]
For each light cheap item $i$, we set $\widetilde w(i)=w(i)$ and $\widetilde p(i)=p(i)$.
For a set of items $X$, we define $\widetilde w(X)=\sum_{i\in X}\widetilde w(i)$ and $\widetilde p(X)=\sum_{i\in X}\widetilde p(i)$.
\end{definition}

\begin{definition}[Type and $\prec$-order]\label{def:prec-light}
    Two light-expensive items $e_1,e_2$ are of the same \emph{type} if
    $\widetilde w(e_1)=\widetilde w(e_2)$ and $\widetilde p(e_1)=\widetilde p(e_2)$.
    We define a partial order $\prec$ on light-expensive items by declaring
    $e_1\prec e_2$ iff $e_1$ and $e_2$ have the same type and $w(e_1)<w(e_2)$.
\end{definition}

Next, we count the number of different types of light-expensive items.
Note that if we choose $W(e)=\varepsilon(1-\widetilde w(t))$, then we have $w(e)\le 1-w(t)$, which implies $\frac{w(e)}{\varepsilon W(e)} \le \frac{1}{\varepsilon^2}$. 
If we choose $W(e)=\varepsilon^3$, then $\frac{w(e)}{\varepsilon W(e)}  \le \frac{1}{\varepsilon^4}$.
On the other hand, $\frac{p(e)}{\varepsilon^2}$ is at most $\frac{1}{\varepsilon^2}$.
We have the following observation.

\begin{observation}[Constant number of types]\label{obs:num_of_types}
    There are at most $O\big((\frac{|T|}{\varepsilon^2}+\frac{1}{\varepsilon^4})\cdot\frac{1}{\varepsilon^2}\big) = O((\frac{1}{\varepsilon})^6)$ different types.
\end{observation}

The following lemma shows that there exists a solution in which, within each type, every critical light item is heavier than every non-critical heavy item (Property~\ref{lem:easy-cl-exist-p4}).
Moreover, the total profit decreases by at most $\varepsilon$ compared to the solution in Lemma~\ref{lem:easy-hc-exist} (Property~~\ref{lem:easy-cl-exist-p1}).
And if we replace the weights of all critical items by their label weights, the feasibility is still preserved (Property~\ref{lem:easy-cl-exist-p2}), where we apply the property of slack.

\begin{restatable}{lemma}{lemeasyclexist}\label{lem:easy-cl-exist}
    Given $H^*_1, \ldots, H^*_m$, there exists a feasible solution $\mathcal{S} = (I_1, \ldots, I_m)$ satisfying
    \begin{enumerate}[label = {\normalfont (\roman*)}]
        \item $p(I_j) \geq \frac{2}{3} - 3\varepsilon$ for any $j \in [m]$.\label{lem:easy-cl-exist-p1}

        \item For $j \in [m]$, the set of critical heavy items in $I_j$ is exactly $H^*_j$.\label{lem:easy-cl-exist-p2}

        \item For each $j \in [m]$, if $|I_j| \geq 3$, then 
        \[
            w(I_j) - w(L^*_j) + \widetilde{w}(L^*_j) - w(H^*_j) + \widetilde{w}(H^*_j) \leq 1,
        \] 
       where $L^*_j$ is the set of critical light items in $I_j$.\label{lem:easy-cl-exist-p3}

       \item For any critical  light-expensive item $e_1$, and for any non-critical light-expensive item $e_2$, if $e_1$ and $e_2$ are of the same type, then $e_1 \succ e_2$.\label{lem:easy-cl-exist-p4}
    \end{enumerate}
\end{restatable}
\begin{proof}
(Deferred to Appendix~\ref{app:easy-cl-exist}.)
\end{proof}

\begin{lemma}\label{lem:cl-guess}
    Given $H^*_1, \ldots, H^*_m$, after guessing for $m^{2^{O(\frac{1}{\varepsilon})}}$ rounds, at least one of the guesses gives $m$ sets $L^*_1, \ldots, L^*_m$ of light-expensive  items such that there exists a feasible solution $\mathcal{S} = (I_1, \ldots, I_m)$ satisfying
    \begin{enumerate}[label = {\normalfont (\roman*)}]
        \item $p(I_j) \geq \frac{2}{3} - 4\varepsilon$ for any $j \in [m]$.\label{lem:cl-guess-p1}

        \item For each $j \in [m]$, the set of critical heavy items of $\mathcal{S}$ in $I_j$ is exactly $H^*_j$.\label{lem:cl-guess-p2}

        \item For each $j \in [m]$, the set of critical light items of $\mathcal{S}$ in $I_j$ is exactly $L^*_j$.\label{lem:cl-guess-p3}

    \end{enumerate}
\end{lemma}
\begin{proof}
    We first guess, for each type of light items, how many critical items of this type appear in the solution $\mathcal{S}'$ from Lemma~\ref{lem:easy-cl-exist}.
    By Observation~\ref{obs:num_of_types}, this can be done in $O\big(n^{(1/\varepsilon)^6}\big)$ rounds.
    By Lemma~\ref{lem:easy-cl-exist}~\ref{lem:easy-cl-exist-p4}, selecting the heaviest items with the guessed number in each type exactly gives the set of critical light items in $\mathcal{S}'$.
    Denote the set by $L^*$.

    We now guess the assignment of critical items in $\mathcal{S}'$ by guessing the assignment of their  types. 
    Let $\mathcal{R}$ be the set of indices of all distinct types of critical items (including both heavy and light items).
    By Observation~\ref{obs:num_of_types} and $|T|=O\big((1/\varepsilon)^2\big)$, we have $|\mathcal{R}|=O\big((1/\varepsilon)^6\big)$.
    Recall that each knapsack contains at most $1/\varepsilon$ critical items.
    A \emph{configuration} is a vector $\mathbf{r}=(r_1,\ldots,r_{1/\varepsilon})\in \mathcal{R}^{1/\varepsilon}$, where $r_k$ specifies which rounded type the $k$-th critical item in the knapsack uses (padding with a dummy type if there are fewer than $1/\varepsilon$ critical items).
    Thus, the number of possible configurations is
    $N\le |\mathcal{R}|^{1/\varepsilon}\le \bigl((1/\varepsilon)^6\bigr)^{1/\varepsilon}
    =2^{O(\frac{1}{\varepsilon}\log\frac{1}{\varepsilon})}$.

    Then we guess the multiplicities of these configurations.
    Let $(c_1,\ldots,c_N)\in\mathbb{Z}_{\ge0}^N$ be the count vector, where $c_i$ denotes the number of knapsacks using configuration $i$.
    We enumerate all possibilities for $(c_1,\ldots,c_N)$ satisfying $\sum_{i=1}^N c_i=m$, which has
    $\binom{m+N-1}{N-1}\le m^{O(N)}=m^{2^{O(1/\varepsilon\log(1/\varepsilon))}}$ possibilities.

    For convenience, an item is said to have a rounded type, determined by its label weight and label profit.
    Note that at least one of the guesses gives the count vector of $\mathcal{S'}$ in Lemma~\ref{lem:easy-cl-exist}. 
    Fixing this guess, we obtain $m$ configurations of $\mathcal{S}'$, say $\mathbf{r}_1,\dots,\mathbf{r}_m$.

    We assign the configurations to knapsacks as follows.
    We compute a perfect matching between the sets $H_1^*,\ldots,H_m^*$ and $\mathbf r_1,\ldots,\mathbf r_m$. 
    We add an edge between $H_j^*$ and $\mathbf r_i$ if either 
    (1) the rounded type of $H_j^*$ appears in $\mathbf r_i$, or 
    (2) $H_j^*=\emptyset$ and $\mathbf r_i$ contains no heavy-item type.
    Note that the perfect matching exists given a correct guess of $\mathcal{S}'$.

    Now we construct $L_1^*,\dots,L_m^*$. We relabel $\mathbf r_1,\ldots,\mathbf r_m$ such that $\mathbf r_j$ matches $H_j^*$. For each rounded type in $\mathbf r_j$, we select an arbitrary unused original item within this type from the guessed critical-item set $L^*$. This selection is valid, because for each type, the total number of occurrences across $\mathbf{r}_1,\dots,\mathbf{r}_m$ equals the guessed number in $L^*$. Then let $L_j^*$ be the set of light items selected for $\mathbf r_j$.

    It remains to show the existence of $\mathcal{S}$ claimed in the Lemma.
    Given $\mathcal{S}'=(I_1',\dots,I_m')$, we replace all the critical items in $I_j'$ with the guessed set $L_j^*$ for all $j\in[m]$, while keeping other items unchanged. 
    Let the resulting solution be $\mathcal{S}=(I_1,\dots,I_m)$. 
    The construction guarantees that Property~\ref{lem:cl-guess-p3} is satisfied.
    By Lemma~\ref{lem:easy-cl-exist}\ref{lem:easy-cl-exist-p2}, Property~\ref{lem:cl-guess-p2} holds.

    Note that $p(I_j)\ge\widetilde p(I_j)=\widetilde p(I_j')$, since we construct $L_j^*$ by selecting light items of the same rounded types as in $\mathbf r_j$.
    It remains to bound $\widetilde p(I_j')$.
    For any light item $l$, we have $p(l)-\widetilde{p}(l)\le \varepsilon^2$ by Definition~\ref{def:label-light}.
    Lemma~\ref{lem:easy-cl-exist}\ref{lem:easy-cl-exist-p1} guarantees that $p(I_j')\ge \frac23-3\varepsilon$ for all $j\in [m]$.
    Since each $I_j'$ contains at most $1/\varepsilon$ critical light items, we have $\widetilde{p}(I_j')\ge p(I_j')-\varepsilon^2\cdot \frac{1}{\varepsilon}\ge \frac{2}{3}-4\varepsilon$, which completes the proof.
\end{proof}

\subsection{Determining Non-Critical Items}
It remains to determine the non-critical items (of $\mathcal{S}$ in Lemma~\ref{lem:cl-guess}) in each knapsack.
There are two kinds of non-critical items: the items in $I_j$ with $|I_j| = 2$ and the cheap items in $I_j$ with $|I_j| \geq 3$. 

\subsubsection{Determining the items in $I_j$ with $|I_j| = 2$.} 

In solution $\mathcal{S}$ of Lemma~\ref{lem:cl-guess}, we have $p(I_j) \geq \frac{2}{3} - 4\varepsilon$ for all $j\in[m]$.
Recall that the profit of each item does not exceed $\frac{2}{3} - 6\varepsilon$.
Therefore, if $|I_j|=2$, each item has a profit at least $(\frac{2}{3} - 4\varepsilon) - (\frac{2}{3} - 6\varepsilon) > \varepsilon$.
We have the following simple observation.

\begin{observation}\label{obs:pair}
    Items in $I_j$ of $\mathcal{S}$ in Lemma~\ref{lem:cl-guess} with $|I_j| = 2$ are expensive.
\end{observation}


The following lemma determines the expensive items in all knapsacks.

\begin{lemma}\label{lem:pair-by-matching-short}
Given $(H_1^*,\ldots,H_m^*)$ and $(L_1^*,\ldots,L_m^*)$, one can compute sets $(E_j^*)_{j\in[m]}$ in time $O(m^{2.5})$
such that there exists a feasible solution $\mathcal{S}=(I_1,\ldots,I_m)$ satisfying
\begin{enumerate}[label = {\normalfont (\roman*)}]
    \item $p(I_j) \ge \frac{2}{3}-4\varepsilon$ for all $j\in[m]$.\label{lem:pair-by-matching-short-p1}
    \item For each $j\in[m]$, the expensive items in $I_j$ are exactly $E_j^*$.\label{lem:pair-by-matching-short-p2}
\end{enumerate}
\end{lemma}
\begin{proof}
    We first compute the sets $(E_j^*)_{j\in[m]}$.
    Let $E$ denote the set of all expensive items and let $NE=I\setminus \bigcup_{j\in[m]} (H_j^*\cup L_j^*)$ denote the set of non-critical expensive items.
    If $H_j^*$ or $L_j^*$ is non-empty, then let $E_j^*=H_j^*\cup L_j^*$.
    It remains to determine $E_j^*$ for knapsacks $j$ with $H_j^*=L_j^*=\emptyset$.

    We construct an undirected graph $G$ as follows. 
    Each item in $NE$ corresponds to a vertex, and we add an edge $\{u,v\}$ iff
    \begin{align}\label{eq:profit-of-matching}
        w(u)+w(v)\le 1
        \quad\text{and}\quad
        p(u)+p(v)\ge \tfrac{2}{3}-4\varepsilon.
    \end{align}
    We then compute a maximum matching $M$ in $G$, which can be done in
    $O(|E|\sqrt{|NE|})=O(m^{2.5})$ time~\cite{hopcroft1973n}.
    Let $n^*$ denote the number of knapsacks $j$ with $H_j^*=L_j^*=\emptyset$.
    We select arbitrary $\min\{|M|,n^*\}$ knapsacks $j$ with $H_j^*=L_j^*=\emptyset$, and set $E_j^*$ to be the pair of items corresponding to an edge in $M$; for all other such knapsacks, we set $E_j^*=\emptyset$.
    It remains to show the existence of the desired solution $\mathcal{S}$.

    Let $\mathcal{S}'=(I_1',\ldots,I_m')$ be the solution guaranteed by Lemma~\ref{lem:cl-guess}, and let
    $Z(\mathcal{S}')=\{j\in[m] : |I_j'|=2\}$.
    For each $j\in Z(\mathcal{S}')$, writing $I_j'  =\{u_j,v_j\}$, Lemma~\ref{lem:cl-guess} implies that
    $p(u_j)+p(v_j)\ge \tfrac{2}{3}-4\varepsilon$ and $w(u_j)+w(v_j)\le 1$.
    By the definition of $G$, this means $\{u_j,v_j\}$ is an edge of $G$.
    Moreover, these edges are pairwise vertex-disjoint (since the $I_j$'s are disjoint), and hence they form a matching of size $|Z(\mathcal{S}')|$ in $G$.
    Therefore, a maximum matching $M$ satisfies $|M|\ge |Z(\mathcal{S}')|$.

    We construct the desired solution $\mathcal{S}=(I_1,\ldots,I_m)$ as follows.
    Given $\mathcal{S}'=(I_1',\ldots,I_m')$, if $H_j^*=L_j^*=\emptyset$ and $E_j^*\neq\emptyset$, then let $I_j=E_j^*$; otherwise, let $I_j=I_j'$.

    We verify that $\mathcal{S}$ satisfies Property~\ref{lem:pair-by-matching-short-p1}.
    If $I_j=E_j^*$, then by equation~\eqref{eq:profit-of-matching}, $p(E_j^*)\ge \frac{2}{3}-4\varepsilon$. Otherwise, $I_j=I_j'$ and by Lemma~\ref{lem:cl-guess}, $p(I_j')\ge \frac{2}{3}-4\varepsilon$.

    We verify that $\mathcal{S}$ satisfies Property~\ref{lem:pair-by-matching-short-p2}.
    If $I_j=E_j^*$, then the property holds trivially. 
    Otherwise, either at least one of $H_j^*$ and $L_j^*$ is non-empty, or $E_j^*=\emptyset$.  
    For the former case, by Lemma~\ref{lem:cl-guess}, the expensive items in $I_j'$ are the critical items $H_j^*\cup L_j^*$, which is exactly $E_j^*$ by our construction.
    For the latter case, it implies that $H_j^*=L_j^*=\emptyset$.
    Thus, the critical (i.e., expensive) items in $I_j'$ are empty, which is exactly $E_j^*$.
\end{proof}

\subsubsection{Determine the cheap items in $I_j$ with $|I_j| \geq 3$}
Let $C=I\setminus \bigcup_{j\in[m]} E_j^*$ denote the set of cheap items. Let the variable $x_{ij} \in \{0, 1\}$ represent whether the cheap item $i$ is packed in knapsack $j$. Relaxing $x_{ij}$ to be fractional, we have the following LP. 
\begin{align}\label{eq:linearprogramming}
    \begin{aligned}
    &\max &&t \\
    &s.t. &&p(E_j^*) + \sum_{i\in C} p(i) x_{ij} \ge t, \quad \forall j=1,\dots,m \\
    &     &&w(E_j^*) + \sum_{i\in C} w(i) x_{ij} \le B, \quad \forall j=1,\dots,m \\
    &     &&\sum_{j=1}^{m} x_{ij} = 1, \quad \forall i\in C \\
    &     &&0 \le x_{ij} \le 1, \quad \forall i \in C,  j=1,\dots,m.
    \end{aligned}
\end{align}

Note that Lemma~\ref{lem:pair-by-matching-short} guarantees a feasible solution to the above LP with objective value at least $\frac{2}{3}-4\varepsilon$. Applying Lemma~\ref{lem:sliding} below, we can obtain an integral feasible solution with objective value at least  $\frac{2}{3}-4\varepsilon-2\max_{i\in C}p(i)$. Since items in $C$ are cheap, this objective value is at least $\frac{2}{3}-6\varepsilon$, which completes the proof of Theorem~\ref{thm:identical}. 

\begin{restatable}{lemma}{lemsliding}\label{lem:sliding}
Given nonnegative $p(i), w(i), t_j, B_j$ and a fractional solution $x_{ij}^f$ to the following linear system:
    \begin{align}\label{eq:linearprogramminga}
    \begin{aligned}
    & &&\sum_{i=1}^n p(i) x_{ij}^f \ge t_j, \quad \forall j=1,\dots,m \\
    &     &&\sum_{i=1}^n w(i) x_{ij}^f \le B_j, \quad \forall j=1,\dots,m \\
    &     &&\sum_{j=1}^{m} x_{ij}^f = 1, \quad \forall i=1,\cdots,n \\
    &     &&0 \le x_{ij}^f \le 1, \quad \forall i =1,\cdots,n, j=1,\dots,m
    \end{aligned}
\end{align}
we can compute in polynomial time an integral solution $x_{ij}^*$ to the following:
\begin{align}\label{eq:linearprogramming2}
    \begin{aligned}
    & &&\sum_{i=1}^n p(i) x_{ij}^* \ge t_j-2p_{\max}, \quad \forall j=1,\dots,m \\
    &     &&\sum_{i=1}^n w(i) x_{ij}^* \le B_j, \quad \forall j=1,\dots,m \\
    &     &&\sum_{j=1}^{m} x_{ij}^* = 1, \quad \forall i=1,\cdots,n \\
    &     && x_{ij}^* \in\{0,1\}, \quad \forall i =1,\cdots,n, j=1,\dots,m
    \end{aligned}
\end{align}
where $p_{\max}=\max_{i}p(i)$.
\end{restatable}
\begin{proof}
(Deferred to Appendix~\ref{app:sliding}.)
\end{proof}




\section{An Algorithm for Arbitrary Capacities}
In this section, we prove Theorem~\ref{thm:general}. 
Recall that $1 = B_1 \leq B_2 \leq \cdots \leq B_m$ and that the optimal objective value $\mathrm{OPT} = 1$. 
We denote an instance of the problem by $\bigl(I,p,w,(B_j)_{j\in[m]}\bigr)$.
\subsection{Roadmap of the proof}
\begin{enumerate}[leftmargin=*, itemsep=2pt, topsep=2pt]
\item \textbf{Rounding the instance.}
We round knapsack capacities and item profits and weights to powers of $1+\varepsilon$.
We show that the rounded instance still admits a feasible solution of value at least $1-\varepsilon$ (Lemma~\ref{Lem:Arb-round-instance}).
Moreover, we give a sufficient condition for converting a rounded solution into a feasible solution of the original instance (Observation~\ref{obs:round-to-original}).
Thus, in the remainder of the proof we work exclusively with the rounded instance.


\item \textbf{Configuration LP.}
A \emph{configuration} characterizes the expensive items packed in a knapsack by their rounded profits and rounded weights.
We introduce a configuration LP whose variables specify how many times each configuration is used across the $m$ knapsacks
(they sum to $m$, possibly fractionally).
Lemma~\ref{lem:config-feasible} shows that this LP admits a feasible solution.
Moreover, we enforce that a constant number of configuration variables take integral values (Lemma~\ref{lem:config-sol}).

\item \textbf{Assignment LP.}
Given a solution to the configuration LP, we derive an assignment LP that explicitly assigns items to knapsacks
(Lemma~\ref{lem:alp-sol}).

\item \textbf{Rounding the assignment LP.}
It remains to round the fractional assignment to an integral one.
We first construct a fractional solution such that the total weight of fractional items is at most $\varepsilon^2$ times the capacity (Lemma~\ref{lem:slack-alp-sol}).
Then we round the remaining fractional solution to an integral one with a loss of at most $1/2$
(Lemma~\ref{lem:expensive-round}).
After fixing the expensive-item choices, cheap items can be solved using essentially the same sliding argument as in
Lemma~\ref{lem:sliding} (Lemma~\ref{lem:cheap-round}).
\end{enumerate}

\subsection{Preliminaries}
Our goal is to compute a solution $(I_1, \ldots, I_m)$ with $p(I_j) \geq \frac{1}{2} - \frac72\varepsilon$ for each $j\in [m]$. For simplicity, we assume that $\varepsilon$ is a sufficiently small constant and $1/\varepsilon$ is an integer.

\smallskip
\noindent\textbf{Assumption on item profits.} We assume that every item $i$ has $p(i) < \frac{1}{2} - \frac72\varepsilon$ since otherwise we can assign $i$ to the smallest knapsack with $B_j \geq w(i)$, and then discard the item and the knapsack from the instance. 

\begin{definition}[Expensive and Cheap Items]\label{def:arb-expensive-cheap-items}
    We say an item $i \in I$ is expensive if $p(i) \geq \varepsilon$, and cheap otherwise. Let $E$ be the set of all expensive items.
\end{definition}

\noindent\textbf{Assumption on number of expensive items in each knapsack.}
We assume that each knapsack contains at most $1/\varepsilon$ expensive items. If a knapsack contains more, we keep an arbitrary subset of $1/\varepsilon$ expensive items, whose total profit is at least $1$.

\begin{definition}[Label weights, label profits and label capacities]\label{Def:Arb-label}
    For any expensive item $i$, we define its label weight $\widetilde{w}(i)$ and its label profit $\widetilde{p}(i)$ as
    \[
        \widetilde{w}(i)=\varepsilon^2 (1 + \varepsilon^2)^v \text{ and } \widetilde{p}(i)=\varepsilon (1 + \varepsilon)^u
    \]
    where $v$ is the smallest integer such that $w(i)\le \varepsilon^2 (1 + \varepsilon^2)^v$ and $u$ is the largest integer such that $p(i)\ge \varepsilon (1 + \varepsilon)^u$.
    For any cheap item $i$, let $\widetilde{w}(i)=w(i)$ and $\widetilde{p}(i)=p(i)$.
    For any knapsack $j\in [m]$, we define the label capacity as
    \[
        \widetilde{B}_j=(1+\varepsilon)^{t+1}
    \]
    where $t$ is the smallest integer such that $B_j\le (1+\varepsilon)^t$.
\end{definition}

For any expensive item $i$, we have $\varepsilon^2 (1 + \varepsilon^2)^{v-1}< w(i)\le \varepsilon^2 (1 + \varepsilon^2)^v$ for some integer $v$ and therefore
\begin{align}\label{eq:arb-weight-bound}
    \widetilde w(i)-w(i)\le \varepsilon^2 (1 + \varepsilon^2)^v-\varepsilon^2 (1 + \varepsilon^2)^{v-1}=\varepsilon^4(1+\varepsilon^2)^{v-1}< \varepsilon^2 w(i).
\end{align}
For profits, we have $\varepsilon (1 + \varepsilon)^u \le p(i)< \varepsilon (1 + \varepsilon)^{u+1}$ for some integer $u$ and therefore
\begin{align}\label{eq:arb-profit-bound}
    p(i)-\widetilde p(i)\le \varepsilon (1 + \varepsilon)^{u+1}-\varepsilon (1 + \varepsilon)^{u}= \varepsilon^2(1 + \varepsilon)^{u}\le \varepsilon p(i)<\varepsilon,
\end{align}
where the last inequality is due to $p(i)<\frac12-\frac27 \varepsilon<1$.
For capacities, 
\begin{align}\label{eq:arb-capacity-bound}
    \widetilde B_j-B_j\ge (1+\varepsilon)^{t+1}-(1+\varepsilon)^{t}=\varepsilon (1+\varepsilon)^{t}\ge \varepsilon B_j.
\end{align}

\subsection{Rounding the Instance}
\begin{lemma}\label{Lem:Arb-round-instance}
    Given instance $\bigl(I,p,w,(B_j)_{j\in[m]}\bigr)$, in $O(|I|)$ time we can compute a feasible rounded instance $\bigl(I, \widetilde p,\widetilde w,(\widetilde B_j)_{j\in[m]}\bigr)$ such that
    there exists a feasible solution with  $\widetilde{p}(I_j)\ge 1-\varepsilon$ for all $j\in[m]$.
\end{lemma}
\begin{proof}
    We round the profits, weights and capacities of $I$ to their label values in $O(|I|)$ time according to Definition~\ref{Def:Arb-label}.
    Let $\mathcal{S}=(I_1,\dots,I_m)$ be an optimal solution to $I$ with objective value $1$.
    We show that, after rounding, $\mathcal{S}$ remains feasible and attains objective value at least $1-\varepsilon$.

    Consider the weight of $\mathcal{S}$ after rounding.
    By equation~\eqref{eq:arb-weight-bound}, the weight of each expensive item $i\in I_j$ increases by at most $\varepsilon^2 w(i)\le \varepsilon^2 B_j$.  
    Therefore, the total weight increase is at most $\frac{1}{\varepsilon}\cdot \varepsilon^2 B_j\le \varepsilon B_j$.
    On the other hand, the capacity of each knapsack increases by at least $\varepsilon B_j$ by equation~\eqref{eq:arb-capacity-bound}. Hence the knapsack remains feasible.

    By the profit-rounding rule in Definition~\ref{Def:Arb-label}, the profit of each expensive item decreases by at most a factor of $(1 + \varepsilon)$.
    Consequently, the objective value becomes at least $\frac{1}{1+\varepsilon} \ge 1-\varepsilon$.
\end{proof}

Given $\widetilde{B}_j=(1+\varepsilon)^q$, by Definition~\ref{Def:Arb-label}, its original weight $B_j$ is at least $(1+\varepsilon)^{q-2}\ge (1+\varepsilon)^{q-2}(1+\varepsilon)^2 (1-2\varepsilon)\ge \widetilde{B}_j (1-2\varepsilon)$. 
Therefore, if we have a rounded solution with $\widetilde{w}(I_j)\le (1-2\varepsilon)\widetilde{B}_j$, this solution is also feasible in the original instance. 
Moreover, by the profit-rounding rule, it holds that $p(I_j)\ge \widetilde{p}(I_j)$ for all $j\in [m]$.
The following Observation gives a sufficient condition for converting a rounded solution to a feasible solution of the original instance.

\begin{observation}\label{obs:round-to-original}
    Let $\bigl(I, \widetilde p,\widetilde w,(\widetilde B_j)_{j\in[m]}\bigr)$ be a rounded instance.
    If there exists a solution $(I_1,\ldots,I_m)$ to the rounded instance with objective value $t$ and
    \[
    \widetilde{w}(I_j)\le (1-2\varepsilon)\,\widetilde{B}_j \qquad \text{for all } j\in[m],
    \]
    then the same assignment $(I_1,\ldots,I_m)$ is feasible for the original instance and achieves objective value at least $t$.
\end{observation}

From now on, we focus on the rounded instance $\bigl(I, \widetilde p,\widetilde w,(\widetilde B_j)_{j\in[m]}\bigr)$ and show that, in polynomial time, we can find a solution $(I_1, \ldots, I_m)$ with $\widetilde{p}(I_j) \geq \frac{1}{2} - 3\varepsilon$ and $\widetilde{w}(I_j) \leq (1 - 2\varepsilon) \widetilde{B}_j$.

\subsection{Configuration LP}
Our goal in this subsection is to construct a configuration LP and show that we can compute a solution of objective value at least $1-\varepsilon$ in which a constant number of variables are integral.
We first introduce notation for item types and knapsack types.

\noindent\textbf{Notation for item types.}
Let $E^u_v$ be the set of expensive items whose profits are $\varepsilon(1 + \varepsilon)^u$ and whose weights are $\varepsilon^2(1 + \varepsilon^2)^v$. Let
\[
U = \{ u \in \mathbb{Z} : \exists i \in E,\ \widetilde{p}(i) = \varepsilon(1+\varepsilon)^u \},
\qquad
V = \{ v \in \mathbb{Z} : \exists i \in E,\ \widetilde{w}(i) = \varepsilon^2(1+\varepsilon^2)^v \}.
\]
be the set of all possible choices of $u$ and $v$ respectively.  
Since we assume that $p(i)< \frac12-\frac27\varepsilon<1$, we have $|U| \leq O(\frac{1}{\varepsilon}\log \frac{1}{\varepsilon})$. Moreover, trivially $|V|\le |E|\le n$. 
We define 
\[
    E^u_{\leq v}=\bigcup_{v' \leq v} E^u_{v'}
\] 
to be the set of items with rounded profit $\varepsilon(1+\varepsilon)^u$ and rounded weight at most $\varepsilon^2(1+\varepsilon^2)^v$.

\smallskip
\noindent\textbf{Notation for knapsack types.}
Let $\mathcal{K}_q$ be the set of knapsacks whose capacities are $(1 + \varepsilon)^q$.
Let
\begin{align}\label{eq:notation-Q}
    Q = \{ q \in \mathbb{Z} : \exists j \in [m],\ \widetilde{B}_j = (1+\varepsilon)^q \},
\end{align}
be the set of all possible choices of $q$.

\smallskip
\noindent\textbf{Item types that fit in a knapsack.}
For a fixed $q$, define 
\[
    v(q) \;=\; \max\Bigl\{ v \in \mathbb{Z} \;:\; \varepsilon^2(1+\varepsilon^2)^v \le (1+\varepsilon)^q \Bigr\}.
\]
Here $v(q)$ is the index of the largest rounded weight that fits into knapsacks of $\mathcal{K}_q$.
To keep the number of item types per knapsack constant, for knapsacks in $\mathcal{K}_q$ we only consider weight indices in the range $[q,v(q)]$.
Note that we have $v(q) - q  = O(\frac{1}{\varepsilon}\log \frac{1}{\varepsilon})$ by the definition.

Next, we define configurations. A configuration is a vector recording the number of selected expensive items of each type.
We use a cumulative representation: for each $u\in U$ and each $v\in[q,v(q)]$, let $n^u_{\le v}\in \mathbb{Z}$ denote the number of selected items from $E^u_{\le v}$.
\begin{definition}[Configurations]
A configuration (of expensive items) for knapsacks in $\mathcal{K}_q$ is a vector
$
C=\bigl\langle n^u_{\le v}\bigr\rangle_{u\in U,\; q\le v\le v(q)}
$
of integers satisfying the following conditions:
\begin{enumerate}[label=(\roman*),leftmargin=*]
    \item For all $u\in U$ and $q\le v\le v(q)$,
    $
        0 \le n^u_{\le v} \le |E^u_{\le v}|.
    $
    \item For all $u\in U$ and $q\le v'\le v\le v(q)$,
    $
        n^u_{\le v'} \le n^u_{\le v}.
    $
\end{enumerate}

\end{definition}

\begin{definition}[Profits and weights of configurations]\label{Def:PWofConfig}
    We define the profit and the weight of $C=\bigl\langle n^u_{\le v}\bigr\rangle_{u\in U,\; q\le v\le v(q)}$ as
    \begin{align}
        \widetilde{p}(C) = 
        &\sum_{u\in U} n^u_{\leq v(q)} \varepsilon(1 + \varepsilon)^u, \nonumber\\
        \widetilde{w}(C) = 
        &\sum_{u\in U} n^u_{\leq q}\varepsilon^2(1 + \varepsilon^2)^{q} + \sum_{u\in U}\sum_{v: q < v \leq v(q) } (n^u_{\leq v} - n^u_{\leq v-1}) \varepsilon^2(1 + \varepsilon^2)^{v}.\nonumber
    \end{align}
\end{definition}
Note that when calculating $\widetilde{w}(C)$, items with weight type below $q$ are considered as type $q$. 

\begin{definition}[Cumulative number of items of a configuration]
    Given a configuration $C$, we define $n^{u}_{\leq v}(C)$ to be a function that 
\[
    n^{u}_{\leq v}(C) = \left\{
    \begin{array}{ll}
        0 & \text{if $v < q$}\\
        n^u_{\leq v} &\text{if $q \leq v \leq v(q)$}\\
        n^u_{\leq v(q)} &\text{if $v > v(q)$}.
    \end{array}
    \right.
\]
\end{definition}

We restrict attention to configurations that select at most $1/\varepsilon$ expensive items and whose total rounded weight fits a knapsack in $\mathcal{K}_q$.
Formally, for each fixed $q$, let

\begin{align}\label{eq:Cq}
    \mathcal{C}_q
    =\Bigl\{\, \bigl\langle n^u_{\le v}\bigr\rangle_{u\in U,\; q\le v\le v(q)} :
    \sum_{u\in U} n^u_{\le v(q)}(C)\le \frac{1}{\varepsilon}
    \ \text{ and }\ 
    \widetilde{w}(C)\le (1+\varepsilon)^q
    \Bigr\}.
\end{align}

\begin{observation}\label{obs:arb-num-of-config}
    For a fixed $q\in Q$, the number of possible configurations $|\mathcal{C}_q| \leq (1/\varepsilon)^{(v(q) - q)|U|} =  (1/\varepsilon)^{(1/\varepsilon)^2\log^2(1/\varepsilon)}$. The number of configurations over all knapsacks is $O(m\cdot |\mathcal{C}_q|)=O_{\varepsilon}(m)$.
\end{observation}

\smallskip
\noindent\textbf{Configuration LP.}
For each $q\in Q$ and each $C\in\mathcal{C}_q$, we introduce a variable $x_C$ that denotes the number of knapsacks that use configuration $C$.
We also consider the assignment of cheap items.
Recall that $E$ is the set of expensive items. For each cheap item $i\in I\setminus E$, we use a variable $y_{iC}$ to indicate whether item $i$ is assigned to configuration $C$.
The constraints can be interpreted as follows:
\begin{enumerate}[label=(\roman*),leftmargin=*]
    \item The first constraint enforces that there are $|\mathcal{K}_q|$ configurations selected for knapsacks in $\mathcal{K}_q$.
    \item The second constraint ensures that the profit of configuration $C$ ($x_C$ units), together with its associated cheap items, attains total profit at least $(1-\varepsilon)x_C$.
    \item The third constraint ensures that the weight of configuration $C$ ($x_C$ units), together with its associated cheap items, does not exceed  $(1+\varepsilon)^{q}x_C$.
    \item The fourth constraint guarantees that, for every $u\in U$ and $v$, the number of expensive items from $E^u_{\le v}$ used by the configurations is at most $|E^u_{\leq v}|$.
    \item The fifth constraint enforces that each cheap item is used at most once.
\end{enumerate}

\begin{align*}
    \sum_{C \in \mathcal{C}_q} x_{C} &= |\mathcal{K}_q| &&\forall q \in Q\\
    \sum_{i \in I\setminus E}\widetilde{p}(i)y_{iC} & \geq (1 - \varepsilon - \widetilde{p}(C))x_{C} &&\forall q \in Q, \forall C \in \mathcal{C}_q\\
    \sum_{i \in I\setminus E}\widetilde{w}(i)y_{iC}  &\leq \left((1 + \varepsilon)^{q} - \widetilde{w}(C)\right)x_{C} &&\forall q \in Q, \forall C \in \mathcal{C}_q\\
    \sum_{j = 1}^m \sum_{C \in \mathcal{C}_j} n^u_{\leq v}(C)x_C & \leq |E^u_{\leq v}|  &&\forall u \in U, v \in V\\
    \sum_{q \in Q} \sum_{C \in \mathcal{C}_q} y_{iC} &\leq 1  &&\forall i \in I\setminus E\\
     x_{C} &\geq  0&& \forall q \in Q, \forall C \in \mathcal{C}_q\\
    0 \leq y_{iC} &\leq 1 && \forall q \in Q, \forall C \in \mathcal{C}_q, \forall i\in I \setminus E
\end{align*}

We show that the solution satisfying Lemma~\ref{Lem:Arb-round-instance} is a feasible solution of this configuration LP with objective value at least $1-\varepsilon$.

\begin{restatable}{lemma}{lemconfigfeasible}\label{lem:config-feasible}
The configuration LP has a feasible solution.
\end{restatable}
\begin{proof}
(Deferred to Appendix~\ref{app:config-feasible}.)
\end{proof}

Moreover, in polynomial time, we can ask for a feasible solution such that a constant number of variables are integral. Let $\eta = \max_{j\in [m]} |\mathcal{C}_j|$.
Let
\begin{align}\label{eq:arb-constant}
    \tau \;=\; \left\lceil \frac{2\eta\log(1/\varepsilon)}{\log(1+\varepsilon)} \right\rceil,
    \quad
    s \;=\; \left\lceil\frac{1}{3\varepsilon}\right\rceil,
    \quad
    \Delta \;=\; \left\lceil \frac{\log\!\bigl(s/(1-2\varepsilon)\bigr)}{\log(1+\varepsilon)} \right\rceil,
    \quad
    L \;=\; \tau s + \Delta.
\end{align}

\noindent\textbf{Remark.} The choice of $\tau,s,\Delta,L$ will be clear in Lemma~\ref{lem:slack-alp-sol}. One may assume $L$ is simply a constant depending only on $\varepsilon$ here.

\begin{lemma}\label{lem:config-sol}
In time $2^{2^{poly(1/\varepsilon)}}\cdot poly(|I|)$,
we can compute a feasible solution for the configuration LP such that for the smallest $L$ elements $q \in Q$,
we have $x_{C} \in \mathbb{Z}$ for every $C \in \mathcal{C}_q$.
\end{lemma}

\begin{proof}
Recall that $\eta=\max_{q\in Q}|\mathcal{C}_q|$. By Observation~\ref{obs:arb-num-of-config},
$\eta$ is bounded by a function of $\varepsilon$ only. Hence, by~\eqref{eq:arb-constant},
$L=\tau s+\Delta$ also depends only on $\varepsilon$.

We consider the optimization problem obtained from the configuration LP by additionally imposing the integrality
constraints $x_C\in\mathbb{Z}$ for all $C\in\mathcal{C}_q$ where $q$ ranges over the smallest $L$ elements of $Q$.
This is a mixed integer linear program (MILP).
Let $Q_{\le L}$ be the set of the smallest $L$ elements of $Q$.
The number of integer variables is at most
\[
k \;\le\; \sum_{q\in Q_{\le L}}|\mathcal{C}_q| \;\le\; L\eta,
\]
which is bounded by a function of $\varepsilon$ only.

Next, we bound the encoding length of this mixed-integer linear program.
The formulation contains $O(m\eta)$ variables $x_C$ and $O(|I|\cdot m\eta)$ variables $y_{iC}$, and it has
$O(m\eta+|U||V|+|I|)$ constraints. Since $|U|=poly(1/\varepsilon)$ and $|V|\le |I|$, and $\eta$ depends only on
$\varepsilon$, the encoding length of the mixed-integer linear program is $poly(|I|\cdot \eta)$.

By Kannan's fixed-dimension integer programming algorithm~\cite{Kan87}, a mixed-integer linear program with
$k$ integer variables can be solved in time $2^{O(k\log k)}\cdot poly(|MILP|)$, where $|MILP|$
denotes the encoding length of the MILP. In our case, $|MILP|=poly(|I|\cdot \eta)$ and
$k$ depends only on $\varepsilon$, hence $2^{O(k\log k)}\le 2^{2^{poly(1/\varepsilon)}}$.
Therefore, the overall running time is at most $2^{2^{poly(1/\varepsilon)}}\cdot poly(|I|)$, as claimed.
\end{proof}

\subsection{Assignment LP}
Given a feasible solution of the configuration LP satisfying Lemma~\ref{lem:config-sol}, we show that it can be converted to a solution of the following assignment LP, where $z_{ij}$ indicates if item $i$ is assigned to knapsack $j$.
\begin{align*}
    \mathrm{s.t.}\qquad \sum_{i \in I} z_{ij}\widetilde{w}(i) &\leq \widetilde{B}_j &&\forall j \in [m]\\
    \sum_{i \in I} z_{ij}\widetilde{p}(i) &\geq (1-\varepsilon) &&\forall j \in [m]\\
    \sum_{j=1}^m z_{ij} &\leq 1  &&\forall i \in I\\
    0 \leq z_{ij} &\leq 1 && \forall j \in [m], \forall i\in I
\end{align*}

\begin{restatable}{lemma}{lemalpsol}\label{lem:alp-sol}
    In $2^{2^{poly(1/\varepsilon)}}\cdot poly(|I|)$ time, we can compute a feasible solution $\{z_{ij}\}_{i\in I, j\in [m]}$ with objective value at least $(1 - \varepsilon)$ for the assignment LP such that
    \begin{enumerate}[label={\normalfont (\roman*)}]
        \item for the smallest $L$ elements $q \in Q$, every $z_{ij} \in \{0,1\}$ for every $i \in E$ and every $j \in \mathcal{K}_q$.\label{lem:alp-sol-p1}

        \item for each $q \in Q$, at most $\eta$ of the knapsacks in $\mathcal{K}_q$ receive fractional expensive items, that is, have non-integral $z_{ij}$ for some $i \in E$. Recall that $\eta = \max_{j\in [m]} |\mathcal{C}_j|$.\label{lem:alp-sol-p2}

        \item for each knapsack $j \in [m]$, we have $\widetilde{w}(i) \leq \widetilde{B}_j$ for every expensive item $i \in E$ with $z_{ij} > 0$.\label{lem:alp-sol-p3}
    \end{enumerate}
\end{restatable}
\begin{proof}[Proof sketch.] We defer the detailed proof to Appendix~\ref{app:alp-sol}. 

Property~\ref{lem:alp-sol-p1} follows from Lemma~\ref{lem:config-sol} and a simple construction:
for the smallest $\tau+\lceil\tau/\varepsilon\rceil$ groups, whenever a knapsack receives an
integral configuration, we assign expensive items integrally.

For Property~\ref{lem:alp-sol-p2}, fix $q\in Q$. Intuitively, for each configuration
$C\in\mathcal C_q$ we can allocate $\lfloor x_C\rfloor$ copies of $C$ to distinct knapsacks
in $\mathcal K_q$ (thus assigning their expensive items integrally). After doing this for all
$C$, at most one partially used copy remains per configuration, hence at most $|\mathcal C_q|\le\eta$
knapsacks can be affected by fractional expensive items.

Property~\ref{lem:alp-sol-p3} holds since knapsacks in $\mathcal K_q$ receive expensive items
from configurations in $\mathcal C_q$; by definition of $\mathcal C_q$ (see equation~\ref{eq:Cq}), any $C\in \mathcal C_q$ satisfies $\widetilde{w}(C)\le (1+\varepsilon)^{q}$, which implies that every item $i$ in the configuration satisfies $\widetilde{w}(i) \leq \widetilde{B}_j$.
\end{proof}

\subsection{Rounding Solutions of Assignment LP}
Given the assignment-LP solution in Lemma~\ref{lem:alp-sol}, we construct a new feasible solution in which every knapsack that receives fractional expensive items uses only a small fraction of its capacity (Lemma~\ref{lem:slack-alp-sol}). 
Then we can round this fractional
solution to an integral one with objective value at least $\frac12-\varepsilon$ (Lemma~\ref{lem:expensive-round}).

\begin{lemma}\label{lem:slack-alp-sol}
In $poly(|I|)+2^{poly(1/\varepsilon)}$ time, we can compute a solution $\{z_{ij}\}_{i\in I, j\in [m]}$ for the assignment LP such that the following holds for each $j \in [m]$.
\begin{enumerate}[label={\normalfont (\roman*)}]
\item If $z_{ij}$ is integral for all $i \in E$, then
\[
    \sum_{i \in I} z_{ij}\widetilde w(i)\leq (1 - 2\varepsilon)\widetilde B_j
    \quad \text{and} \quad
    \sum_{i \in I} z_{ij}\widetilde p(i) \geq \frac{1}{2} - \varepsilon.
\]\label{lem:slack-alp-sol-p1}
\item If $0 < z_{ij} < 1$ for some $i \in E$, then
\[
    \sum_{i \in I} z_{ij}\widetilde w(i) \leq \varepsilon^2 \widetilde B_j
    \quad \text{and} \quad
    \sum_{i \in I} z_{ij}\widetilde p(i) \geq 1 - \varepsilon.
\]
Moreover, $\widetilde w(i) \leq \varepsilon^2 \widetilde B_j$ for any $i \in E$ with $z_{ij} > 0$.\label{lem:slack-alp-sol-p2}
\end{enumerate}
\end{lemma}

\begin{proof}
Let $(z_{ij})_{i\in I,\,j\in[m]}$ be a feasible solution returned by Lemma~\ref{lem:alp-sol}.
For each knapsack $j\in[m]$, define the set of assigned items
\[
    I_j \;=\; \{\, (i, z_{ij}) : i\in I,\ z_{ij}>0 \,\}.
\]
For any subset $S\subseteq I_j$, define
\[
    \widetilde p(S) = \sum_{(i,\alpha)\in S} \alpha\,\widetilde p(i),
    \qquad
    \widetilde w(S) = \sum_{(i,\alpha)\in S} \alpha\,\widetilde w(i),
\]
and $\rho(i)=\widetilde p(i)/\widetilde w(i)$.
We view each pair $(i,\alpha)\in I_j$ as a single item of profit $\alpha\widetilde p(i)$ and weight $\alpha\widetilde w(i)$.

\begin{claim}\label{clm:removed-set-2}
Fix $j\in[m]$.
Assume that $\widetilde p(I_j)\ge 1-\varepsilon$ and that every item $i$ with $z_{ij}>0$ satisfies
$\widetilde p(i)<\frac12-\frac72 \varepsilon$.
Then in $poly(|I|)$ time we can compute a subset $D_j\subseteq I_j$ such that
\[
    3\varepsilon\,\widetilde p(I_j) \le \widetilde p(D_j) < \frac12\,\widetilde p(I_j)
    \qquad\text{and}\qquad
    \widetilde w(D_j) \ge 3\varepsilon\,\widetilde w(I_j).
\]
\end{claim}

\begin{proof}
Sort all items in $I_j$ in non-decreasing order of density.
Scan them in this order and keep adding items to $D_j$ until
$\widetilde p(D_j)\ge 3\varepsilon\,\widetilde p(I_j)$ holds for the first time.
Let $(i^*,\alpha^*)$ be the last added item. Then
$\widetilde p(D_j\setminus\{(i^*,\alpha^*)\})<3\varepsilon\,\widetilde p(I_j)$.

Since $\alpha^*\le 1$ and $\widetilde p(i^*)<\frac12-\frac72\varepsilon$, we have
$\alpha^*\widetilde p(i^*) \le \widetilde p(i^*) < \frac12-\frac72\varepsilon$.
Therefore,
\[
\begin{aligned}
\widetilde p(D_j)
&= \widetilde p\bigl(D_j\setminus\{(i^*,\alpha^*)\}\bigr) + \alpha^* \widetilde p(i^*) \\
&< 3\varepsilon\,\widetilde p(I_j) + \Bigl(\tfrac12-\tfrac72\varepsilon\Bigr)
< 3\varepsilon\,\widetilde p(I_j) + \Bigl(\tfrac12-3\varepsilon\Bigr)\widetilde p(I_j)
= \tfrac12\,\widetilde p(I_j),
\end{aligned}
\]
where we used $\widetilde p(I_j)\ge 1-\varepsilon$ in the second inequality.
Together with the stopping condition, this gives
$3\varepsilon\,\widetilde p(I_j)\le \widetilde p(D_j) < \frac12\,\widetilde p(I_j)$.

Moreover, since $D_j$ consists of the minimum-density items in $I_j$, its average density is at most that of $I_j$:
$\frac{\widetilde p(D_j)}{\widetilde w(D_j)} \le \frac{\widetilde p(I_j)}{\widetilde w(I_j)}$.
Rearranging yields $\widetilde w(D_j)\ge 3\varepsilon\,\widetilde w(I_j)$.
The sorting and selecting procedure runs in $poly(|I|)$ time.
\end{proof}

Partition knapsacks according to whether they receive a fractional expensive item:
\begin{align}
  \mathcal{J}^{\mathrm{int}}&
    = \bigl\{\, j\in[m] : z_{ij}\in\{0,1\}\text{ for all }i\in E \,\bigr\}, \nonumber\\
    \mathcal{J}^{\mathrm{frac}}&
    = \bigl\{\, j\in[m] : \exists\, i\in E \text{ with } 0<z_{ij}<1 \,\bigr\}.\nonumber
\end{align}
These two sets form a partition of $[m]$.

Fix any $j\in \mathcal{J}^{\mathrm{int}}$.
Feasibility of the assignment LP implies $\widetilde p(I_j)\ge 1-\varepsilon$ and $\widetilde w(I_j)\le \widetilde B_j$.
By the assumption on item profits, every $i$ with $z_{ij}>0$ satisfies
$\widetilde p(i)<\frac12-\frac72\varepsilon$.
Hence Claim~\ref{clm:removed-set-2} guarantees a set $D_j\subseteq I_j$ with $\widetilde p(D_j)<\frac12\widetilde{p}(I_j)$ and $\widetilde w(D_j)\ge 3\varepsilon \widetilde w(I_j)$.
Let
\[
    I'_j=I_j\setminus D_j.
\]
Then
\[
    \widetilde p(I'_j)
    = \widetilde p(I_j)-\widetilde p(D_j)
    \ge \tfrac12\,\widetilde p(I_j)
    > \tfrac12-\varepsilon,
\]
and
\[
    \widetilde w(I'_j)
    = \widetilde w(I_j)-\widetilde w(D_j)
    \le (1-3\varepsilon)\widetilde w(I_j)
    \le (1-3\varepsilon)\widetilde B_j
    \le (1-2\varepsilon)\widetilde B_j.
\]
Since all expensive variables are integral in $I_j$, they remain integral in $I'_j$.
Thus knapsack $j$ satisfies~\ref{lem:slack-alp-sol-p1}.

Now consider the knapsacks in $\mathcal{J}^{\mathrm{frac}}$ and order them by nondecreasing label capacity:
\[
    \mathcal{J}^{\mathrm{frac}}=\{j_1,\ldots,j_k\},
    \qquad
    \widetilde B_{j_1}\le \widetilde B_{j_2}\le \cdots \le \widetilde B_{j_k}.
\]
For each $r\in[k]$, let $q(r)\in Q$ be such that $j_r\in \mathcal{K}_{q(r)}$.
By Lemma~\ref{lem:alp-sol}, each $\mathcal{K}_q$ contains at most $\eta$ indices from $\mathcal{J}^{\mathrm{frac}}$.
Recall that in equation~\eqref{eq:arb-constant} we defined
\[
    \tau \;=\; \left\lceil \frac{2\eta\log(1/\varepsilon)}{\log(1+\varepsilon)} \right\rceil .
\]

\begin{claim}\label{clm:large-capacity-diff}
For every $r$ with $1\le r\le k-\tau$, we have
\[
    \widetilde B_{j_r} \;\le\; \varepsilon^{2}\,\widetilde B_{j_{r+\tau}}.
\]
\end{claim}

\begin{proof}
Fix $r\in\{1,\ldots,k-\tau\}$ and consider $j_r,\ldots,j_{r+\tau}$.
Since knapsacks in $\mathcal{K}_q$ appear at most $\eta$ times in $\mathcal{J}^{\mathrm{frac}}$,
among these $\tau+1$ knapsacks there are at least $\left\lceil \frac{\tau}{\eta}\right\rceil$ distinct label capacities.
As label capacities are powers of $(1+\varepsilon)$, we obtain
\[
    \widetilde B_{j_{r+\tau}}
    \;\ge\; (1+\varepsilon)^{\tau/\eta}\,\widetilde B_{j_r}.
\]
By the definition of $\tau$ (see equation~\eqref{eq:arb-constant}), we have $(1+\varepsilon)^{\tau/\eta}\ge (1/\varepsilon)^2$, and thus
$\widetilde B_{j_r}\le \varepsilon^2 \widetilde B_{j_{r+\tau}}$.
\end{proof}

Define a new family of sets $(I'_j)_{j\in[m]}$ as follows.
For each $r\in\{\tau+1,\ldots,k\}$, set $I'_{j_r}=I_{j_{r-\tau}}$, and set $I'_{j_1}=\cdots=I'_{j_\tau}=\emptyset$.
Fix any $r\in\{\tau+1,\ldots,k\}$. Then
\[
    \widetilde p(I'_{j_r})=\widetilde p(I_{j_{r-\tau}})\ge 1-\varepsilon,
\]
and by Claim~\ref{clm:large-capacity-diff}
\[
    \widetilde w(I'_{j_r})
    =\widetilde w(I_{j_{r-\tau}})
    \le \widetilde B_{j_{r-\tau}}
    \le \varepsilon^2\,\widetilde B_{j_r}.
\]
Moreover, for any expensive item $i\in E$ with $z_{i,j_{r-\tau}}>0$, Lemma~\ref{lem:alp-sol}(ii) gives
$\widetilde w(i)\le \varepsilon^2 \widetilde B_{j_{r-\tau}}$, and since $\widetilde B_{j_{r-\tau}}\le \widetilde B_{j_r}$ we also have
$\widetilde w(i)\le \varepsilon^2 \widetilde B_{j_r}$.
Therefore each knapsack $j_{\tau+1},\ldots,j_k$ satisfies~\ref{lem:slack-alp-sol-p2}.

It remains to handle the empty knapsacks $j_1,\ldots,j_\tau$.
Recall that in~\eqref{eq:arb-constant} we defined
\[
    s \;=\; \left\lceil\frac{1}{3\varepsilon}\right\rceil,
    \qquad
    \Delta \;=\; \left\lceil \frac{\log\!\bigl(s/(1-2\varepsilon)\bigr)}{\log(1+\varepsilon)} \right\rceil,
    \qquad
    L \;=\; \tau s + \Delta .
\]
By Lemma~\ref{lem:alp-sol-p1} with the parameter $L$, for every knapsack
$j\in \bigcup_{q\in Q_{\le L}}\mathcal{K}_q$ and every expensive item $i\in E$ we have $z_{ij}\in\{0,1\}$, where
$Q_{\le L}\subseteq Q$ denotes the set of the $L$ smallest indices in $Q$.
Let $Q_{\le \tau s}\subseteq Q_{\le L}$ be the set of the $\tau s$ smallest indices in $Q_{\le L}$.
By equation~\eqref{eq:notation-Q}, for any $q\in Q$, we have $\mathcal{K}_q$ is nonempty. Hence we can choose $\tau s$ \emph{distinct} knapsacks from
\[
    \bigcup_{q\in Q_{\le \tau s}} \mathcal{K}_q;
\]
denote the chosen knapsacks by $u_1,u_2,\ldots,u_{\tau s}$.

For each $\ell\in[\tau s]$, the knapsack $u_\ell$ belongs to $\mathcal{J}^{\mathrm{int}}$, hence the set
$D_{u_\ell}$ from Claim~\ref{clm:removed-set-2} is well-defined and satisfies
\[
    \widetilde p(D_{u_\ell}) \ge 3\varepsilon\,\widetilde p(I_{u_\ell}) \ge 3\varepsilon(1-\varepsilon),
    \qquad
    \widetilde w(D_{u_\ell}) \le \widetilde w(I_{u_\ell}) \le \widetilde B_{u_\ell}.
\]

Next we compare the capacities of $u_\ell$ and $j_1$.
All knapsacks in $\bigcup_{q\in Q_{\le L}}\mathcal{K}_q$ have integral expensive assignments, while $j_1$ receives a fractional
expensive item by definition of $\mathcal{J}^{\mathrm{frac}}$. Let $q_1$ be the index such that $j_1\in\mathcal{K}_{q_1}$.
In particular, we have $q_1> L=\tau s+\Delta$.
Therefore, the index $q_1$ is at least $\Delta$ larger than every index in $Q_{\le \tau s}$, and hence
\[
    \widetilde B_{u_\ell}
    \le (1+\varepsilon)^{-\Delta}\,\widetilde B_{j_1}
    \le \frac{1-2\varepsilon}{s}\,\widetilde B_{j_1}
    \le \frac{1-2\varepsilon}{s}\,\widetilde B_{j_r}
    \qquad\text{for all } \ell\in[\tau s],\ r\in[\tau],
\]
where we used $\widetilde B_{j_1}\le \widetilde B_{j_r}$.

We now fill the empty knapsacks by grouping these sets $D_{u_\ell}$ into $\tau$ blocks, where each block contains $s$ sets.
Specifically, for each $r\in[\tau]$, define
\[
    I'_{j_r} \;=\; \bigcup_{t=1}^{s} D_{u_{(r-1)s+t}}.
\]
Then for each $r\in[\tau]$, by Claim~\ref{clm:removed-set-2}
\[
    \widetilde p(I'_{j_r})
    \ge s\cdot 3\varepsilon(1-\varepsilon)
    \ge 1-\varepsilon
    \ge \frac12-\varepsilon,
\]
and
\[
    \widetilde w(I'_{j_r})
    \le \sum_{t=1}^{s}\widetilde w\!\left(D_{u_{(r-1)s+t}}\right)
    \le \sum_{t=1}^{s}\widetilde B_{u_{(r-1)s+t}}
    \le s\cdot \frac{1-2\varepsilon}{s}\,\widetilde B_{j_r}
    = (1-2\varepsilon)\widetilde B_{j_r}.
\]
Moreover, every expensive item contained in $I'_{j_r}$ is integral because each $D_{u_\ell}$ is taken from a knapsack
whose expensive assignments are integral. Hence each of $j_1,\ldots,j_\tau$ satisfies~\textnormal{(i)}.

Altogether, every knapsack satisfies either~\textnormal{(i)} or~\textnormal{(ii)}. Moreover, the running time of the construction is polynomial in the number of variables of $(z_{ij})_{i\in I,\,j\in[m]}$ and $L$, which gives a total running time of $poly(|I|+L)=poly(|I|)+2^{poly(1/\varepsilon)}$.
\end{proof}

Lemma~\ref{lem:slack-alp-sol}\textnormal{(ii)} implies that if knapsack $j$ contains a fractional expensive item, then its current load is tiny:
$\sum_{i} z_{ij}\widetilde w(i)\le \varepsilon^2\widetilde B_j$, while its profit is already large,
$\sum_i z_{ij}\widetilde p(i)\ge 1-\varepsilon$.
Since each such knapsack contains at most $1/\varepsilon$ expensive items and each of them has weight at most $\varepsilon^2\widetilde B_j$,
rounding the expensive variables cannot violate the capacity constraint; thus, for these knapsacks we may essentially ignore weights and only track profits.
Then the problem can be viewed as (a special case of) the unrelated machine scheduling problem where knapsacks are machines and expensive items are jobs with processing times
$\widetilde p(i)$.
By an argument of Lenstra, Shmoys and Tardos~\cite{LST87}, we can transform the solution so that each knapsack has at most one fractional expensive
variable, without introducing any new expensive item into a knapsack.
Rounding the remaining fractional variable then decreases the profit of a knapsack by at most $\max_{i\in E}\widetilde p(i)\le 1/2$,
yielding a minimum profit of at least $\frac12-\varepsilon$.

\begin{lemma}\label{lem:expensive-round}
    In $\mathrm{poly}(|I|)$ time, we can compute a feasible solution with objective value at least $(\frac{1}{2} - \varepsilon)$ for the assignment LP such that
    \begin{enumerate}[label={\normalfont (\roman*)}]
        \item $z_{ij} \in \{0,1\}$ for every $i \in E$ and every $j \in [m]$;

        \item for each $j \in [m]$, 
        \[
            \sum_{i \in I} z_{ij}\widetilde w(i)\leq (1 - 2\varepsilon) \widetilde B_j. 
        \]
    \end{enumerate}
\end{lemma}
\begin{proof}
Let $(z_{ij})_{i\in I,\,j\in[m]}$ be a feasible solution of the assignment LP satisfying Lemma~\ref{lem:slack-alp-sol}.
Recall the partition from Lemma~\ref{lem:slack-alp-sol}:
\[
    \mathcal{J}^{\mathrm{frac}}
    = \{\, j\in[m] : \exists\, i\in E \text{ with } 0<z_{ij}<1 \,\},
    \qquad
    \mathcal{J}^{\mathrm{int}}
    = [m]\setminus \mathcal{J}^{\mathrm{frac}}.
\]
Let
\[
    I' = \{\, i\in E : \exists\, j\in \mathcal{J}^{\mathrm{frac}} \text{ with } z_{ij}>0 \,\}
\]
be the set of expensive items that appear in knapsacks of $\mathcal{J}^{\mathrm{frac}}$.
For each $j\in \mathcal{J}^{\mathrm{frac}}$, let
\[
    P_j = \sum_{i\in I\setminus E} z_{ij}\,\widetilde p(i)
\]
be the profit contributed by cheap items in knapsack $j$ (this quantity will be kept unchanged).
For each $(i,j)\in I'\times \mathcal{J}^{\mathrm{frac}}$, define
\[
    p_{ij} =
    \begin{cases}
        \widetilde p(i) & \text{if } z_{ij}>0,\\
        0 & \text{if } z_{ij}=0,
    \end{cases}
\]
so that item $i$ may only be (re)assigned to knapsacks that already contain a positive fraction of $i$.

Consider the following linear program over variables $(x_{ij})_{i\in I',\,j\in \mathcal{J}^{\mathrm{frac}}}$:
\begin{align*}
    \sum_{i\in I'} p_{ij}\,x_{ij} &\ge 1-\varepsilon-P_j && \forall j\in \mathcal{J}^{\mathrm{frac}},\\
    \sum_{j\in \mathcal{J}^{\mathrm{frac}}} x_{ij} &\le 1 && \forall i\in I',\\
    x_{ij} &\ge 0 && \forall i\in I',\ \forall j\in \mathcal{J}^{\mathrm{frac}}.
\end{align*}
Restricting $(z_{ij})$ in the assignment LP to $I'\times \mathcal{J}^{\mathrm{frac}}$ yields a feasible solution to this LP. Hence this LP is feasible.
By an argument of Lenstra, Shmoys, and Tardos~\cite{LST87}, when we focus on an extreme-point solution of our assignment LP (which has the same form as in Theorem~1 of~\cite{LST87}), we can compute a feasible solution $(x_{ij})_{i\in I',\,j\in \mathcal{J}^{\mathrm{frac}}}$ such that each knapsack $j\in \mathcal{J}^{\mathrm{frac}}$ has at most one index $i\in I'$ with
$0<x_{ij}<1$ in time polynomial in the number of variables, namely $poly(|I|)$.

Now define $z'_{ij}$ by keeping all cheap variables unchanged, and replacing the expensive variables on
$I'\times \mathcal{J}^{\mathrm{frac}}$ with $(x_{ij})$:
\[
    z'_{ij} =
    \begin{cases}
        z_{ij} & \text{if } i\notin I',\\
        x_{ij} & \text{if } i\in I' \text{ and } j\in \mathcal{J}^{\mathrm{frac}},\\
        z_{ij} & \text{if } i\in I' \text{ and } j\in \mathcal{J}^{\mathrm{int}}.
    \end{cases}
\]
Finally, for each $j\in \mathcal{J}^{\mathrm{frac}}$, if there is an expensive item $i$ with $0<z'_{ij}<1$, we set $z'_{ij}=0$
(i.e., delete this single fractional expensive item from knapsack $j$).
After this operation, all expensive variables are integral.
The resulting assignment $(z'_{ij})$ is our desired solution.
We now verify its objective value and feasibility.

For any $j\in \mathcal{J}^{\mathrm{frac}}$, the LP constraints ensure that
$P_j + \sum_{i\in I'} p_{ij}z'_{ij} \ge 1-\varepsilon$ before the deletion step, and the deletion step removes the profit of at most
one expensive item. Since $\widetilde p(i)\le \tfrac12$ for every $i\in E$, the profit of knapsack $j$ after deletion is at least
$1-\varepsilon-\tfrac12 = \tfrac12-\varepsilon$.
For $j\in \mathcal{J}^{\mathrm{int}}$, we did not change any expensive variable, and Lemma~\ref{lem:slack-alp-sol}\textnormal{(i)}
already guarantees profit at least $\tfrac12-\varepsilon$.
Hence the resulting assignment attains objective value at least $\tfrac12-\varepsilon$.

It remains to verify the capacity constraints.
For $j\in \mathcal{J}^{\mathrm{int}}$, Lemma~\ref{lem:slack-alp-sol}\textnormal{(i)} gives
$\sum_i z'_{ij}\widetilde w(i)\le (1-2\varepsilon)\widetilde B_j$.
For $j\in \mathcal{J}^{\mathrm{frac}}$, deleting an item can only decrease weight, so it suffices to bound the weight before deletion.
Lemma~\ref{lem:slack-alp-sol}\textnormal{(ii)} implies that every expensive item $i\in E$ with $z'_{ij}>0$ satisfies
$\widetilde w(i)\le \varepsilon^2 \widetilde B_j$.
Moreover, by definition of expensive items we have $\widetilde p(i)\ge \varepsilon$, hence to obtain profit at least
$\tfrac12-\varepsilon$ knapsack $j$ needs at most $1/\varepsilon$ expensive items. Therefore,
\[
    \sum_{i\in I} z'_{ij}\widetilde w(i)
    \le \frac{1}{\varepsilon}\cdot \varepsilon^2 \widetilde B_j
    = \varepsilon \widetilde B_j
    \le (1-2\varepsilon)\widetilde B_j,
\]
for all sufficiently small $\varepsilon$.

The running time is polynomial in the number of variables of the LP, which is at most $poly(|I|)$.
\end{proof}

Finally, we round the fractional cheap items by applying Lemma~\ref{lem:sliding}.

\begin{lemma}\label{lem:cheap-round}
    In $\mathrm{poly}(|I|)$ time, we can compute a feasible solution with objective value at least $\frac{1}{2} - \frac72\varepsilon$ for the assignment LP such that
    \begin{enumerate}[label={\normalfont (\roman*)}]
        \item $z_{ij} \in \{0,1\}$ for every $i \in I$ and every $j \in [m]$;

        \item for each $j \in [m]$, 
        \[
            \sum_{i \in I} z_{ij}\widetilde w(i)\leq (1 - 2\varepsilon)\widetilde B_j. 
        \]
    \end{enumerate}
\end{lemma}
\begin{proof}
Starting from the fractional assignment-LP solution given by Lemma~\ref{lem:expensive-round}, all variables
$z_{ij}$ with $i\in E$ are integral and for every $j\in[m]$ we have
$\sum_{i\in I} z_{ij}\widetilde w(i)\le (1-2\varepsilon)\widetilde B_j$.

Fix the assignments of items in $E$ and consider the remaining variables for cheap items
$I\setminus E$. Let the residual capacity be
$B_j'=(1-2\varepsilon)\widetilde B_j-\sum_{i\in E} z_{ij}\widetilde w(i)$ and the residual profit requirement be
$t_j'=(\frac12-\varepsilon)-\sum_{i\in E} z_{ij}\widetilde p(i)$.
Then $(z_{ij})_{i\in I\setminus E}$ is a feasible fractional solution to the linear system
\begin{align}
    \begin{aligned}
    & &&\sum_{i\in I\setminus E} \widetilde p(i) z_{ij} \ge t_j', \quad \forall j=1,\dots,m \\
    &     &&\sum_{i\in I\setminus E} \widetilde w(i) z_{ij} \le B_j', \quad \forall j=1,\dots,m \\
    &     &&\sum_{j=1}^{m} z_{ij} = 1, \quad \forall i\in I\setminus E \\
    &     &&0 \le z_{ij} \le 1, \quad \forall i\in I\setminus E, j=1,\dots,m. \nonumber
    \end{aligned}
\end{align}

Applying Lemma~\ref{lem:sliding} yields an integral assignment for all cheap items while keeping the weight constraints,
and decreases the profit of each knapsack by at most $2p_{\max}$, where
$p_{\max}=\max_{i\in I\setminus E}\widetilde p(i)\le \varepsilon$ (by Definition~\ref{def:arb-expensive-cheap-items}).
Therefore each knapsack attains profit at least $(\frac12-\varepsilon)-2\varepsilon\ge\frac12-\frac72\varepsilon$,
and the capacity bound $\sum_{i\in I} z_{ij}\widetilde w(i)\le (1-2\varepsilon)\widetilde B_j$ is preserved.

By Lemma~\ref{lem:sliding}, the running time is polynomial in the number of variables, which is at most $poly(|I|)$.
\end{proof}

\section{Hardness Result}
We prove Theorem~\ref{thm:hardness} in this section. 
We will give a reduction from the restricted numerical 3-dimensional matching problem (RN3DM), which is known to be strongly NP-complete~\cite{yu2004minimizing}.  

We first describe the numerical 3-dimensional matching problem (N3DM). Given are three multisets of integers $U,V,W$ where $|U|=|V|=|W|=n$ together with an integer $t$, the goal is to find a subset $M\subseteq U\times V\times W$ such that every element in $U\cup V\cup W$ occurs in $M$ exactly once, and for every $(u,v,w)\in M$, $u+v+w=t$. 

The restricted numerical 3-dimensional matching problem further requires that $V=W=\{1,2,\cdots,n\}$ and $U\subseteq [O(n)]$. The following is an equivalent statement.

\begin{definition}[RN3DM]
    Given a multiset $U=\{u_1,\dots,u_n\}$ of nonnegative integers and an integer $t$ such that $\sum_{j=1}^n u_{j}+n(n+1)=nt$, does there exist two permutations $v$ and $w$ of $\{1,\dots,n\}$ such that 
$$
u_{j}+v(j)+w(j)=t, \text{ for }j=1,\dots,n?
$$
\end{definition}

Now we are ready to prove Theorem~\ref{thm:hardness}.
\begin{proof}[Proof of Theorem~\ref{thm:hardness}]
    Given an instance of RN3DM with $U=\{u_1,\dots,u_n\}$ and $t$, we construct an instance of BMKP as follows.
    There are \(2n\) items and $n$ knapsacks. Every item has a profit of 1. For $i=1,2,\cdots,n$, there are exactly two items of weight $i$.
    The capacity of knapsack \(j \in [n]\) is  \(B_j = t - u_j\). 

    We first show that if the answer to the given RN3DM instance is "yes", then the constructed BMKP instance admits a feasible solution with objective value $2$. This follows by observing that we can pack two items into each knapsack $j$, one of weight $v(j)$ and one of weight $w(j)$. The fact that $v$ and $w$ are permutations of $[n]$ guarantees that we have packed exactly two items of weight $i$ for every $i=1,2,\cdots,n$.


    Next, we show that if the constructed BMKP instance admits a feasible solution of objective value $2$, then the answer to the given RN3DM instance is "yes". Since there are $n$ knapsacks and $2n$ items, each of profit 1, it is clear that in this feasible solution all items are packed, and every knapsack contains exactly two items. Moreover, observe that the total weight of all items is exactly $2\cdot (1+2+\cdots+n)=n(n+1)=nt-\sum_{j=1}^nu_j=\sum_{j=1}^nB_j$, thus the total weight of the two items in knapsack $j$ is exactly $B_j=t-u_j$.
    
     Now we construct the two permutations $v$ and $w$ for the RN3DM instance based on the assignment of items. Specifically, consider the weights of the two items in knapsack $j$. We will show that it is possible to set one weight as $v(j)$ and the other weight as $w(j)$ such that $v$ and $w$ are both permutations. As we showed above that the total weight of items in knapsack $j$ is exactly $B_j$, $v(j)+w(j)=t-u_j$ follows directly.  
     
     It remains to show which of the two item weights in knapsack $j$ should be set to $v(j)$ (and the other is $w(j)$) so that $v$ and $w$ are permutations. We present an algorithm for achieving this. For ease of description, we take a graphical view. We create a graph $G=(Q,E)$ as follows. Each item is represented as a node, so there are in total $2n$ nodes. When we set the weight of an item as $v(j)$ (or $w(j)$), we color the node corresponding to this item red (or blue). We may also abuse the notation by saying we color an item red or blue. The two items in the same knapsack are called partners. We create two types of edges: (i) an edge between the two nodes corresponding to partners, (ii) an edge between two nodes whose corresponding items have the same weight. 
     
     Consider graph $G$. The degree of every node is either $1$ or $2$. Consider an arbitrary node $q\in Q$ of degree 1, and let $(q,q')$ be the edge incident to it. Then the two items corresponding to $q$ and $q'$ are in the same knapsack, and have the same weight. We color $q$ red and $q'$ blue, and then remove the two nodes from graph $G$.
     
     After removing all the degree 1 nodes, let $\bar{G}$ be the residual graph, and $\bar{G}_1,\bar{G}_2,\cdots,\bar{G}_{\ell}$ be its connected components. Since the degree of every node in each $\bar{G}_j$ is exactly 2, $\bar{G}_j$ is a cycle. Let the cycle be $(q_1,q_2,\cdots,q_k)$. We observe the following.
     
     \begin{observation}
        $k=|\bar{G}_j|$ is even for every $j$.
        \end{observation} 
        \begin{proof}
        Suppose, on the contrary, that $k$ is odd.  Without loss of generality, let $q_1,q_2$ be two items of the same weight, then $q_2,q_3$ are partners, and $q_3,q_4$ are two items of the same weight, etc. We get that $q_{k-1},q_k$ are partners, and $q_k,q_1$ have the same weight. Then $q_1,q_2,q_k$ all have the same weight. Since there are at most 2 items of the same weight, we get $k=2$, which is a contradiction.
        \end{proof}
     
    Since $k$ is even, for each cycle $(q_1,q_2,\cdots,q_k)$, we color $q_1,q_3,\cdots,q_{k-1}$ red and $q_2,q_4,\cdots,q_k$ blue.  
    
    Now we have colored every node either red or blue; it remains to obtain the permutation $v$ and $w$ from this coloring. According to our coloring, every edge is between a red node and a blue node, which means: (i) Partners (items in the same knapsack) always have different colors. (ii) Items of the same weight always have different colors. Thus, the $n$ red nodes correspond to $n$ items on $n$ knapsacks, and their weights are exactly $1,2,\cdots,n$. Hence, define $v(j)$ (or $w(j)$) as the weight of red (or blue) node in knapsack $j$, $v(j)$ and $w(j)$ are permutations.

  The above argument shows that the answer to the given RN3DM instance is "yes" if and only if the constructed BMKP instance admits a feasible solution with objective value $2$. 

    Observe that, for our constructed BMKP instance, the objective value is at most $2$. 
    Recall that the instance contains $2n$ items and $n$ knapsacks, and every item has profit $1$.
    If some knapsack contains at least three items, then the remaining $n-1$ knapsacks together contain at most $2n-3$ items. 
    By the pigeonhole principle, at least one knapsack contains at most one item, and thus the objective value is at most $1$.
    Otherwise, every knapsack contains at most two items, and thus the objective value is at most $2$.
    In either case, the objective value is at most $2$.

  Now suppose that BMKP admits a polynomial-time $(\tfrac12+\varepsilon)$-approximation algorithm for some constant $\varepsilon>0$.
  Given an RN3DM instance, we construct the corresponding BMKP instance and run this algorithm.
   If the answer to the RN3DM instance is "yes", then $\mathrm{OPT}=2$ for the BMKP instance and the algorithm returns a solution of value at least
    $(\tfrac12+\varepsilon)\cdot 2 >1$, which exactly computes the optimal solution.
    If the answer is "no", then $\mathrm{OPT}=1$ and the returned value is at most $1$.
    Thus, the algorithm would allow us to decide RN3DM in polynomial time, completing the proof of Theorem~\ref{thm:hardness}.
\end{proof}

\section{Conclusion and Discussion}
We studied the bottleneck multiple knapsack problem (BMKP) under both identical-capacity and arbitrary-capacity settings.
For identical capacities, we presented a polynomial-time $(\frac{2}{3}-\varepsilon)$-approximation algorithm, which nearly matches the known $(\frac{2}{3}+\varepsilon)$ inapproximability bound~\cite{caprara2000multiple}.
For arbitrary capacities, we established a nearly tight guarantee at $\frac{1}{2}$ by giving a $(\frac{1}{2}-\varepsilon)$-approximation algorithm together with a $(\frac{1}{2}+\varepsilon)$-inapproximability result.

Several intriguing questions remain open.
First, although the $\frac{2}{3}$ approximation ratio is known to be tight for the bottleneck multiple subset sum problem with identical capacities,
it remains unclear whether a $\frac{2}{3}$-approximation is achievable when the knapsack capacities of the subset sum problem are arbitrary.

Our hardness result in Theorem~\ref{thm:hardness} shows that, for arbitrary capacities,
no $(\tfrac{1}{2}+\varepsilon)$-approximation algorithm exists for any constant $\varepsilon>0$, unless $\mathrm{P}=\mathrm{NP}$,
even in the restricted case where all item profits are $1$.
This naturally raises the question of whether a $\tfrac{1}{2}$-approximation can be achieved in this unit-profit setting.

Another interesting direction is to study intermediate models between identical and arbitrary capacities.
In particular, it is open whether a $(\frac{2}{3}-\varepsilon)$-approximation is achievable when the number of distinct knapsack capacities is bounded by a constant.

\bibliographystyle{plainurl}
\bibliography{ref}

\appendix
\section{Proofs of Identical Capacities}

\subsection{Proof of Lemma~\ref{lem:easy-hc-exist}}\label{app:easy-hc-exist}
\lemeasyhcexist*
\begin{proof}

We first show that the solution $\mathcal{S}$ guaranteed by Lemma~\ref{lem:T-exist} already satisfies
property~\ref{lem:easy-hc-exist-p2}.
By Definition~\ref{def:profile}, for each $j\in[m]$, if $I_j$ contains a critical heavy item $h$, then there exists
$t\in T$ with $\lfloor p(t)/\varepsilon\rfloor=\lfloor p(h)/\varepsilon\rfloor$ such that
\[
w(I_j)-w(h)\le (1-\varepsilon)(1-w(t)).
\]
In particular, $w(t)\ge w(h)$.
Let $t'\in T$ be the lightest item satisfying $\lfloor p(t')/\varepsilon\rfloor=\lfloor p(h)/\varepsilon\rfloor$ and
$w(t')\ge w(h)$. By definition, $\widetilde w(h)=w(t')$ and hence $w(t)\ge w(t')$.
Therefore,
\[
w(I_j)-w(h)\le (1-\varepsilon)(1-w(t))\le (1-\varepsilon)(1-w(t'))=(1-\varepsilon)(1-\widetilde w(h)).
\]

We next enforce~\ref{lem:easy-hc-exist-p3} by swapping heavy items within each type.
Call a pair $(h,h')$ a \emph{bad pair} if $h$ is critical, $h'$ is non-critical, $h$ and $h'$ have the same type,
and $h\prec h'$ (equivalently, $w(h)<w(h')$).
Fix a type $\tau$ and suppose it contains at least one bad pair.
Let $h^*$ be a $\prec$-minimal critical heavy item of type $\tau$, and let $\bar h^*$ be a $\prec$-maximal non-critical heavy item of type $\tau$ (see Definition~\ref{def:prec-order}).
Then $h^*\prec \bar h^*$, so $(h^*,\bar h^*)$ is also a bad pair.

Let $I_i$ and $I_j$ be the knapsacks containing $h^*$ and $\bar h^*$, respectively, and define the swapped knapsacks
\[
I_i' = (I_i\setminus\{h^*\})\cup\{\bar h^*\},
\qquad
I_j' = (I_j\setminus\{\bar h^*\})\cup\{h^*\}.
\]
Since $h^*$ and $\bar h^*$ have the same type, we have $\widetilde p(h^*)=\widetilde p(\bar h^*)$. Therefore, replacing one by the other does not change the label profit of the heavy item in the knapsack.
For any heavy item $h$, we have $p(h)-\widetilde p(h)\le \varepsilon$.
Therefore, condition~(i) is preserved (up to an $\varepsilon$ loss).

We now verify that property~\ref{lem:easy-hc-exist-p2} is preserved.
Since $w(h^*)\le w(\bar h^*)$, knapsack $j$ only becomes lighter, i.e., $w(I_j')\le w(I_j)$.
For knapsack $i$, using $\widetilde w(h^*)=\widetilde w(\bar h^*)$ we obtain
\[
w(I_i')-w(\bar h^*)=w(I_i)-w(h^*)\le (1-\varepsilon)(1-\widetilde w(h^*))
=(1-\varepsilon)(1-\widetilde w(\bar h^*)).
\]
Thus, property~\ref{lem:easy-hc-exist-p2} continues to hold after the swap.

Finally, we repeat the above operation for each type while a bad pair exists.
Each swap decreases the number of bad pairs of that type.
By choosing the $\prec$-minimal critical item and the $\prec$-maximal non-critical item,
neither of the two involved items will be used in a later swap.
At termination, no bad pair remains, which is exactly~\ref{lem:easy-hc-exist-p3}.
Meanwhile, each swap guarantees that properties~\ref{lem:easy-hc-exist-p1} and~\ref{lem:easy-hc-exist-p2} hold.

\end{proof}

\subsection{Proof of Lemma~\ref{lem:easy-cl-exist}}\label{app:easy-cl-exist}
\lemeasyclexist*
\begin{proof}
It suffices to show that the slack solution $\mathcal{S}$ from Lemma~\ref{lem:hc-guess} also satisfies property~\ref{lem:easy-cl-exist-p3}.
The remaining properties follow by the same swapping argument as in Lemma~\ref{lem:easy-hc-exist}, which preserves
properties~\ref{lem:easy-cl-exist-p2} and~\ref{lem:easy-cl-exist-p4} while incurring at most an additional $\varepsilon$ loss in~\ref{lem:easy-cl-exist-p1}.

    Since $W(e)$ (see Definition~\ref{def:label-light}) is the same for all critical light items in the same knapsack, we denote it by $W$.
    Moreover, each $I_j$ contains at most $1/\varepsilon$ expensive items, and hence $|L_j^*|\le 1/\varepsilon$.
    Therefore,
    \begin{align}\label{eq:once}
        \widetilde{w}(L_j^*)-w(L_j^*)
        \le |L_j^*|\cdot \varepsilon W
        \le W.
    \end{align}
    For $I_j$ that contains only light items,
    \begin{align}
        w(I_j)+\widetilde{w}(L_j^*)-w(L_j^*)\le w(I_j)+W\le1-\varepsilon^3+\varepsilon^3\le 1,\nonumber
    \end{align}
    where the second inequality follows from the fact that solution $\mathcal{S}$ is slack (see Definition~\ref{Def:SlackSolution}).

    For $I_j$ that contains the heavy item in   $H_j^*$, we have that
    \begin{align}
        &w(I_j)+\widetilde{w}(L_j^*)-w(L_j^*)+ \widetilde{w}(H^*_j)-w(H^*_j)\nonumber \\
        \le &(1-\varepsilon)(1-\widetilde{w}(H_j^*))+w(H_j^*)+W+ \widetilde{w}(H^*_j)-w(H^*_j)\nonumber\\
        \le &(1-\varepsilon)(1-\widetilde{w}(H_j^*))+\varepsilon(1-\widetilde{w}(H_j^*))+\widetilde{w}(H^*_j)\nonumber\\
        \le & 1\nonumber.
    \end{align}
    The first inequality is due to Lemma \ref{lem:hc-guess}\ref{lem:hc-guess-p2} and equation~\eqref{eq:once}. The second inequality is due to $W\le \varepsilon(1-\widetilde{w}(H_j^*))$, which follows from the fact that, for any light item $e$ in $I_j$, $w(e)\le 1-w(H_j^*)$.
\end{proof}

\subsection{Proof of Lemma~\ref{lem:sliding}}\label{app:sliding}
\lemsliding*
\begin{proof}
We interpret LP~\eqref{eq:linearprogramming} as follows. There are $n$ items indexed by $[n]$;
each item $i$ has profit $p(i)$ and weight $w(i)>0$.
There are $m$ knapsacks with capacities $B_1,\dots,B_m$.
Let $\mathbf x^f=\{x^f_{ij}\}$ be a feasible fractional assignment such that every knapsack $j$
obtains total profit at least $t_j$ and total weight at most $B_j$.
Write
\[
p_j^*=\sum_{i\in[n]} p(i)x^f_{ij},\qquad
w_j^*=\sum_{i\in[n]} w(i)x^f_{ij},
\]
so $p_j^*\ge t_j$ and $w_j^*\le B_j$.
Let $p_{\max}=\max_{i\in[n]} p(i)$.
We will compute an integral assignment $(T_1,\dots,T_m)$ such that
$w(T_j)\le B_j$ and $p(T_j)\ge t_j-2p_{\max}$ for all $j$.

We process knapsacks in the order $1,2,\dots,m$. Let $\mathtt C_1=[n]$ be the original item set.
Before processing knapsack $j$, let $\mathtt C_j$ denote the set of items that have not been assigned
to any of the knapsacks $1,2,\dots,j-1$.
We will compute an integral set $T_j\subseteq \mathtt C_j$ for knapsack $j$ and then update
$\mathtt C_{j+1}=\mathtt C_j\setminus T_j$.
The construction of $T_j$ is based on a sliding window over the density-sorted order of $\mathtt C_j$.

Fix the current set $\mathtt C_j$ and sort its items in non-increasing order of density
$\rho_i=p(i)/w(i)$. Let $\pi(1),\dots,\pi(|\mathtt C_j|)$ be this order, and define prefix weights
$W[0]=0$ and $W[r]=\sum_{s=1}^r w(\pi(s))$ for $r=1,\dots,|\mathtt C_j|$.
Define the density function $f_j:\mathbb R_{\ge0}\to\mathbb R_{\ge0}$ by
\[
f_j(x)=
\begin{cases}
\rho_{\pi(r)} & \text{if } W[r-1]\le x < W[r],\\
0 & \text{if } x\ge W[|\mathtt C_j|].
\end{cases}
\]
For any interval $[y_1,y_2]$, define its profit as $\int_{y_1}^{y_2} f_j(t)\,dt$.
Now define
\[
F_j(x)=\int_x^{x+w_j^*} f_j(t)\,dt,
\]
the profit of the window of length $w_j^*$ starting at position $x$.

\begin{claim}\label{cl:WindowExists}
Suppose there exists a fractional assignment of items in $\mathtt C_j$ to knapsack $j$ with total
weight $w_j^*$ and total profit $p_j^*$.
Then there exists a window $[y_j,y_j+w_j^*]$ on $f_j$ whose profit equals $p_j^*$, i.e.,
$F_j(y_j)=p_j^*$.
Moreover, after sorting $\mathtt C_j$ by density, such a $y_j$ can be found in $O(|\mathtt C_j|)$ time.
\end{claim}
\begin{proof}
Consider any fractional selection of items from $\mathtt C_j$ with total weight $w_j^*$.
Since items are divisible, the maximum possible profit among all such selections is achieved by
taking the densest available fraction of total weight $w_j^*$, which corresponds exactly to the
prefix interval $[0,w_j^*]$ on $f_j$. Therefore, $F_j(0)\ge p_j^*$.

Also, for any $x\ge W[|\mathtt C_j|]$ we have $F_j(x)=0$, hence in particular $F_j(W[|\mathtt C_j|])=0$.
The function $F_j(x)$ is continuous in $x$ because it is the integral of $f_j$ over a sliding interval.
Since $F_j(0)\ge p_j^*$ and $F_j(W[|\mathtt C_j|])=0\le p_j^*$, by the intermediate value theorem
there exists $y_j\in[0,W[|\mathtt C_j|]]$ such that $F_j(y_j)=p_j^*$.

To compute such a $y_j$, note that $f_j$ is a step function with breakpoints at the prefix weights $W[r]$.
Consequently, $F_j$ is piecewise-linear, and its breakpoints occur only when either endpoint of the
window crosses a breakpoint of $f_j$, i.e., at points of the form $W[r]$ and $W[r]-w_j^*$.
Evaluating $F_j$ at these $O(|\mathtt C_j|)$ candidate points and locating an interval on which $F_j$
crosses value $p_j^*$ allows us to find $y_j$ in $O(|\mathtt C_j|)$ time after sorting.
\end{proof}

Fix such a $y_j$ and consider the window $[y_j,y_j+w_j^*]$ on $f_j$.
This window intersects a consecutive block of items in the density order, say
$\hat T_j=\{t_\ell,t_\ell+1,\dots,t_r\}$, where only the first and last items may be taken fractionally.
We define
\[
T_j=\{t_\ell+1,\dots,t_r-1\},
\]
the set of whole items strictly inside the window.
Since the window has total weight $w_j^*$, we have $w(T_j)\le w_j^*\le B_j$.
Since the window has profit $p_j^*$ and it differs from $T_j$ by removing at most two boundary items,
each of profit at most $p_{\max}$, we have
\[
p(T_j)\ge p_j^*-2p_{\max}\ge t_j-2p_{\max}.
\]
Thus it remains to argue that after removing $T_j$, we can continue the procedure on the residual instance.

\begin{claim}\label{cl:ExchageScheme}
Let $\mathbf x=\{x_{ik}\}_{i\in\mathtt C_j,\ k\ge j}$ be a feasible fractional assignment on $\mathtt C_j$
such that for every $k\ge j$,
\[
\sum_{i\in\mathtt C_j} p(i)x_{ik}=p_k^*,\qquad
\sum_{i\in\mathtt C_j} w(i)x_{ik}=w_k^*,\qquad
\sum_{k\ge j} x_{ik}\le 1\ \ \forall i\in\mathtt C_j.
\]
Then there exists a feasible fractional assignment
$\mathbf x'=\{x'_{ik}\}_{i\in\mathtt C_{j+1},\ k\ge j+1}$ on $\mathtt C_{j+1}=\mathtt C_j\setminus T_j$
such that for every $k\ge j+1$,
\[
\sum_{i\in\mathtt C_{j+1}} p(i)x'_{ik}=p_k^*,\qquad
\sum_{i\in\mathtt C_{j+1}} w(i)x'_{ik}=w_k^*,\qquad
\sum_{k\ge j+1} x'_{ik}\le 1\ \ \forall i\in\mathtt C_{j+1}.
\]
\end{claim}
\begin{proof}
Assume items in $\mathtt C_j$ are indexed in non-increasing order of density.
Let $L=\{1,2,\dots,t_\ell-1\}$ and $R=\{t_r+1,\dots,|\mathtt C_j|\}$ be the items strictly to the left/right
of the window block in this order.
Define the total weight and profit that knapsack $j$ draws from $L$ and $R$:
\[
w_l=\sum_{i\in L} w(i)x_{ij},\quad p_l=\sum_{i\in L} p(i)x_{ij},\qquad
w_r=\sum_{i\in R} w(i)x_{ij},\quad p_r=\sum_{i\in R} p(i)x_{ij}.
\]
Since $w(i)>0$ for all items, $w_l=0$ holds if and only if knapsack $j$ uses no item in $L$.
Define
\[
\bar\rho_l=\begin{cases} p_l/w_l & \text{if }w_l>0,\\ 0 & \text{if }w_l=0,\end{cases}
\qquad
\bar\rho_r=\begin{cases} p_r/w_r & \text{if }w_r>0,\\ 0 & \text{if }w_r=0.\end{cases}
\]
In addition, for the case $w_l=0<w_r$ we will need the tight upper bound on densities available in $R$:
\[
\rho_{\max}(R)=\max\{\rho_i : i\in R,\ x_{ij}>0\}.
\]

For each later knapsack $k>j$, let
\[
W_k=\sum_{h\in T_j} w(h)x_{hk},\qquad
P_k=\sum_{h\in T_j} p(h)x_{hk}.
\]
After removing $T_j$, we set $x'_{hk}=0$ for all $h\in T_j$ and $k>j$, which creates a deficit of
exactly $(W_k,P_k)$ in knapsack $k$. We will compensate this deficit using fractions that were
assigned to knapsack $j$ from $L$ and $R$.

We seek nonnegative numbers $(\eta_k,\xi_k)$ for each $k>j$ such that
\begin{equation}\label{eq:eta-xi-final}
\eta_k+\xi_k=W_k,\qquad
\eta_k\bar\rho_l+\xi_k\bar\rho_r=P_k,
\end{equation}
and the total weights drawn from $L$ and $R$ do not exceed the available budgets:
\begin{equation}\label{eq:budget-final}
\sum_{k>j}\eta_k\le w_l,\qquad \sum_{k>j}\xi_k\le w_r.
\end{equation}
If such numbers exist, we can construct $\mathbf x'$ by taking weight $\eta_k$ from the pool of
fractions $\{x_{ij}:i\in L\}$ and weight $\xi_k$ from the pool $\{x_{ij}:i\in R\}$ (proportionally within each pool),
which preserves both the weight and profit of every knapsack $k>j$ and maintains $\sum_{k>j}x'_{ik}\le 1$.

Solving \eqref{eq:eta-xi-final} (when $\bar\rho_l\neq \bar\rho_r$) gives
\[
\eta_k=\frac{P_k-W_k\bar\rho_r}{\bar\rho_l-\bar\rho_r},\qquad
\xi_k=\frac{W_k\bar\rho_l-P_k}{\bar\rho_l-\bar\rho_r}.
\]
Therefore, $\eta_k,\xi_k\ge0$ follows once $\frac{P_k}{W_k}\in[\bar\rho_r,\bar\rho_l]$.
When $w_l>0$ and $w_r>0$, this holds because every item in $T_j$ has density between the densities
available on the left and on the right, hence
\[
\bar\rho_l \ \ge\ \frac{P_k}{W_k}\ \ge\ \bar\rho_r \qquad\text{for all }k>j \text{ with } W_k>0.
\]

It remains to prove the budget constraints \eqref{eq:budget-final}.
Let $\alpha_h=\sum_{k>j}x_{hk}$ for $h\in T_j$. Since items are not used by earlier knapsacks,
$\sum_{k\ge j}x_{hk}\le 1$ implies $\alpha_h\le 1-x_{hj}$. Hence
\[
\sum_{k>j}P_k=\sum_{h\in T_j} p(h)\alpha_h \le \sum_{h\in T_j} p(h)(1-x_{hj}),
\qquad
\sum_{k>j}W_k=\sum_{h\in T_j} w(h)\alpha_h \le \sum_{h\in T_j} w(h)(1-x_{hj}).
\]
Using $p(T_j)\le p_j^*$ and $\sum_{i\in\mathtt C_j} p(i)x_{ij}=p_j^*$ gives
\[
\sum_{h\in T_j} p(h)(1-x_{hj}) \le p_j^*-\sum_{h\in T_j} p(h)x_{hj}
= \sum_{i\in L} p(i)x_{ij}+\sum_{i\in R} p(i)x_{ij}
= p_l+p_r,
\]
and similarly,
\[
\sum_{h\in T_j} w(h)(1-x_{hj}) \le w_j^*-\sum_{h\in T_j} w(h)x_{hj}
= \sum_{i\in L} w(i)x_{ij}+\sum_{i\in R} w(i)x_{ij}
= w_l+w_r.
\]
Combining these inequalities yields
\[
\sum_{k>j}P_k\le p_l+p_r,\qquad \sum_{k>j}W_k\le w_l+w_r.
\]
Summing \eqref{eq:eta-xi-final} over all $k>j$ gives two linear equations in
$\sum_{k>j}\eta_k$ and $\sum_{k>j}\xi_k$, and solving for $\sum_{k>j}\eta_k$ yields
\[
\sum_{k>j}\eta_k
=\frac{\sum_{k>j}P_k-\bar\rho_r\sum_{k>j}W_k}{\bar\rho_l-\bar\rho_r}
\le \frac{(p_l+p_r)-\bar\rho_r(w_l+w_r)}{\bar\rho_l-\bar\rho_r}
= w_l,
\]
and similarly $\sum_{k>j}\xi_k\le w_r$. Thus \eqref{eq:budget-final} holds.

The remaining boundary cases can now be handled directly.
If $w_r=0$, then \eqref{eq:budget-final} forces $\xi_k=0$ for all $k>j$, so we compensate using only $L$.
If $w_l=0$ and $w_r>0$, then by $w(i)>0$ we have $L=\emptyset$ and \eqref{eq:budget-final} forces
$\eta_k=0$ for all $k>j$. In this situation, feasibility requires that every deficit $(W_k,P_k)$
can be compensated using only pieces from $R$, whose densities are all at most $\rho_{\max}(R)$.
Hence we must have $P_k\le W_k\rho_{\max}(R)$ for all $k>j$; otherwise compensation is impossible.
On the other hand, since every item in $T_j$ has density at least $\rho_{\max}(R)$, we have
$P_k\ge W_k\rho_{\max}(R)$. Therefore $P_k=W_k\rho_{\max}(R)$, and setting $\eta_k=0$ and $\xi_k=W_k$
satisfies \eqref{eq:eta-xi-final} and \eqref{eq:budget-final}.
This completes the construction of $\mathbf x'$ and proves the claim.
\end{proof}

Applying Claim~\ref{cl:WindowExists} to knapsack $j$ yields a window of profit $p_j^*$ and weight $w_j^*$,
from which we define $T_j$ as above. Then $w(T_j)\le B_j$ and $p(T_j)\ge t_j-2p_{\max}$.
Claim~\ref{cl:ExchageScheme} guarantees that after removing $T_j$ we can continue with knapsack $j+1$.
Repeating for $j=1,2,\dots,m$ completes the proof of Lemma~\ref{lem:sliding}.
\end{proof}

\clearpage

\section{Proofs of Arbitrary Capacities}
\subsection{Proof of Lemma~\ref{lem:config-feasible}}\label{app:config-feasible}
\lemconfigfeasible*

\begin{proof}
By Lemma~\ref{Lem:Arb-round-instance}, the rounded instance
$\widetilde{\mathcal I}=(I,\widetilde p,\widetilde w,(\widetilde B_j)_{j\in[m]})$
admits an assignment $(I_1,\dots,I_m)$ such that for every $j\in[m]$,
$\widetilde p(I_j)\ge 1-\varepsilon$ and $\widetilde w(I_j)\le \widetilde B_j$.
According to the assumption, each $I_j$ contains at most $1/\varepsilon$ expensive items.

We now define an integral solution $(x,y)$ to the configuration LP.
Fix a knapsack $j$ and let $q$ be the integer satisfying $\widetilde B_j=(1+\varepsilon)^q$.
For every $u\in U$ and $q\le v\le v(q)$, let 
\[
n^u_{\le v}(j)=|I_j\cap E^u_{\le v}|,
\]
and let $C(j)$ be the configuration $\langle n^u_{\le v}(j)\rangle_{u\in U,\;q\le v\le v(q)}$.
Initially, let all variables $(x,y)$ to be $0$. For each $q\in Q$ and each $C\in\mathcal C_q$, set
\[
x_C \;=\; \bigl|\{\,j\in\mathcal K_q : C(j)=C\,\}\bigr|.
\]
Namely, $x_C$ equals the number of knapsacks in $\mathcal K_q$ that use configuration $C$.
For each cheap item $i\in I\setminus E$, if $i\in I_j$ and $C(j)=C$, set $y_{iC}=1$; otherwise
set $y_{iC}=0$. Clearly, all $x_C$ and $y_{iC}$ are integral.

Recall that for each knapsack $j$, the number of expensive items is 
\[
\sum_{u\in U} n^u_{\le v(q)}(j)=|I_j\cap E|\le 1/\varepsilon.
\]
Moreover, by Definition~\ref{Def:PWofConfig} and equation~(\ref{eq:arb-weight-bound}), when calculating $\widetilde w(C(j))$, the weight of each item in $I_j\cap E$ increases at most $\varepsilon^2 (1+\varepsilon)^q= \varepsilon^2 \widetilde{B}_j$ compared to its original weight.
By equation~(\ref{eq:arb-capacity-bound}), we have
\begin{align}\label{eq:config-feasible-1}
    \widetilde w(C(j))+\widetilde{w}(I_j \setminus E)\le B_j+\frac{1}{\varepsilon} \cdot \varepsilon^2 \widetilde{B}_j\le \widetilde{B}_j.
\end{align}

We now verify feasibility. The first constraint holds because every knapsack in $\mathcal K_q$
contributes exactly one unit to exactly one configuration in $\mathcal C_q$.
The fourth constraint holds since each expensive item is assigned to at most one knapsack, and thus,
for every $(u,v)$, the total number of selected items from $E^u_{\le v}$ used by all configurations,
$\sum_{q\in Q}\sum_{C\in\mathcal C_q} n^u_{\le v}(C)x_C$, is at most $|E^u_{\le v}|$.
The fifth constraint holds because each cheap item belongs to at most one set $I_j$.

It remains to check the second and third constraints. Fix any $q\in Q$ and $C\in\mathcal C_q$.
If $x_C=0$ then both constraints are trivial. Otherwise, consider any knapsack $j\in\mathcal K_q$
with $C(j)=C$. By construction,
\[
\sum_{i\in I\setminus E}\widetilde p(i)y_{iC}
= \widetilde p(I_j\setminus E)
= \widetilde p(I_j)-\widetilde p(I_j\cap E)
\ge (1-\varepsilon)-\widetilde p(C),
\]
and by equation~\ref{eq:config-feasible-1},
\[
\sum_{i\in I\setminus E}\widetilde w(i)y_{iC}
= \widetilde w(I_j\setminus E)
= \widetilde w(I_j)-\widetilde w(I_j\cap E)
\le \widetilde B_j-\widetilde w(C)
= (1+\varepsilon)^q-\widetilde w(C).
\]
Summing these inequalities over all knapsacks $j\in\mathcal K_q$ that satisfy $C(j)=C$
yields exactly the second and third constraints for $(q,C)$.
Therefore, $(x,y)$ is a feasible integral solution to the configuration LP, as claimed.
\end{proof}

\subsection{Proof of Lemma~\ref{lem:alp-sol}}\label{app:alp-sol}
\lemalpsol*
\begin{proof}
We obtain a feasible solution $\{x_C\}_{q\in Q,C\in\mathcal C_q}$ and
$\{y_{iC}\}_{q\in Q,C\in\mathcal C_q,i\in I\setminus E}$ to the configuration LP satisfying
Lemma~\ref{lem:config-sol}. We convert it into a feasible solution
$\{z_{ij}\}_{i\in I, j\in [m]}$ of the assignment LP. We first describe the construction and then
verify feasibility and properties (i)--(iii).

(\emph{The construction.})
We first assign configurations to knapsacks. Fix $q\in Q$.
We define variables $\{x_{Cj}\}_{C\in\mathcal C_q,\ j\in\mathcal K_q}$, where $x_{Cj}\ge 0$
denotes the amount of configuration $C$ assigned to knapsack $j$, and we enforce
\begin{align}
\sum_{C\in\mathcal C_q} x_{Cj} &= 1 &&\forall j\in\mathcal K_q, \label{eq:colsum}\\
\sum_{j\in\mathcal K_q} x_{Cj} &= x_C &&\forall C\in\mathcal C_q. \label{eq:rowsum}
\end{align}
By considering these equations as a max-flow problem, it is easy to see that a solution exists.
Consider the following bipartite flow network. Create a source node $s$, a node for each
configuration $C\in\mathcal C_q$, a node for each knapsack $j\in\mathcal K_q$, and a sink node $t$.
Add an arc $(s,C)$ of capacity $x_C$ for every $C\in\mathcal C_q$, an arc $(j,t)$ of capacity $1$
for every $j\in\mathcal K_q$, and an arc $(C,j)$ of capacity $+\infty$ for every pair
$C\in\mathcal C_q$ and $j\in\mathcal K_q$.
We send a flow of value $|\mathcal K_q|$ from $s$ to $t$.
Let $x_{Cj}$ be the amount of flow on arc $(C,j)$.
Note that the equation~\ref{eq:colsum} and~\ref{eq:rowsum} exactly describe the flow constraints of each node.
It is clear that this flow problem has a feasible solution with flow $|\mathcal{K}_q|$ since
\[
\sum_{C\in\mathcal C_q} x_C = |\mathcal K_q|
= \sum_{j\in\mathcal K_q} 1.
\]

We solve equations~\ref{eq:colsum} and~\ref{eq:rowsum} and choose $\{x_{Cj}\}$ to be an \emph{extreme point} (basic feasible solution) of this programming.

Next we convert $\{x_{Cj}\}$ into an assignment of expensive items.
For each $j\in\mathcal K_q$, each type $u\in U$ and each $v\in V$, define
\[
n^u_{\le v}(j) = \sum_{C\in\mathcal C_q} x_{Cj}\, n^u_{\le v}(C),
\]
and let
\[
\Delta^{u}_{v}(j)= n^{u}_{\le v}(j)-n^{u}_{\le (v-1)}(j),
\]
with the convention $n^u_{\le(q-1)}(j)=0$.

We now define the fractions $\alpha^{u,v}_{ij}$ that allocate items in $E^u_v$ to knapsacks.
Let
\[
\mathcal J^{\mathrm{int}}_q = \{j\in\mathcal K_q : \exists\,C\in\mathcal C_q \text{ with } x_{Cj}=1\},
\qquad
\mathcal J^{\mathrm{frac}}_q = \mathcal K_q\setminus \mathcal J^{\mathrm{int}}_q.
\]
For each pair $(u,v)$, we first satisfy the demands of knapsacks in $\mathcal J^{\mathrm{int}}_q$ by
assigning integral items:
for every $j\in\mathcal J^{\mathrm{int}}_q$, pick an arbitrary subset
$S^{u}_{v}(j)\subseteq E^u_v$ of cardinality $\Delta^{u}_{v}(j)$ (this is well-defined since
$\Delta^{u}_{v}(j)$ is an integer for $j\in\mathcal J^{\mathrm{int}}_q$), set
$\alpha^{u,v}_{ij}\leftarrow 1$ for $i\in S^{u}_{v}(j)$ and $\alpha^{u,v}_{ij}\leftarrow 0$ otherwise,
and remove all items in $\bigcup_{j\in\mathcal J^{\mathrm{int}}_q}S^{u}_{v}(j)$ from the pool.
Let $\overline E^{u}_{v}$ denote the remaining items.

For the remaining knapsacks $j\in\mathcal J^{\mathrm{frac}}_q$, we choose
$\{\alpha^{u,v}_{ij}\}_{i\in\overline E^{u}_{v},\, j\in\mathcal J^{\mathrm{frac}}_q}$ to satisfy
\begin{align}
\sum_{i\in \overline E^{u}_{v}} \alpha^{u,v}_{ij} &= \Delta^{u}_{v}(j) && \forall j\in\mathcal J^{\mathrm{frac}}_q, \label{eq:alpha-demand}\\
\sum_{j\in \mathcal J^{\mathrm{frac}}_q} \alpha^{u,v}_{ij} &\le 1 && \forall i\in \overline E^{u}_{v}, \label{eq:alpha-supply}\\
0\le \alpha^{u,v}_{ij} &\le 1 && \forall i\in \overline E^{u}_{v},\ \forall j\in\mathcal J^{\mathrm{frac}}_q. \label{eq:alpha-box}
\end{align}
Such a solution exists because the total remaining demand does not exceed the remaining supply:
\[
\sum_{j\in\mathcal J^{\mathrm{frac}}_q}\Delta^{u}_{v}(j)
\le
\sum_{j\in\mathcal K_q}\Delta^{u}_{v}(j)
=
\sum_{C\in\mathcal C_q} x_C\bigl(n^u_{\le v}(C)-n^u_{\le(v-1)}(C)\bigr)
\le |E^{u}_{v}|,
\]
where the equality uses \eqref{eq:rowsum} and the inequality follows from the fourth constraint of the
configuration LP (applied to $v$ and $v-1$ and taking the difference).
Among all feasible solutions to \eqref{eq:alpha-demand}--\eqref{eq:alpha-box},
we choose an extreme point solution.
For every expensive item $i\in E$ and every knapsack $j\in\mathcal K_q$, define
\[
z_{ij}\ =\ \sum_{(u,v):\, i\in E^{u}_{v}} \alpha^{u,v}_{ij}.
\]
(For $j\notin\mathcal K_q$, we set $z_{ij}=0$ for all $i\in E$.)

Next, we assign cheap items and define $z_{ij}$ for $i\in I\setminus E$.
For each configuration $C$ with $x_C>0$, the variable $y_{iC}$ denotes the total fraction of cheap item
$i$ used across all $x_C$ units of configuration $C$; hence each unit of $C$ carries $y_{iC}/x_C$ of item $i$.
Accordingly, for every $j\in\mathcal K_q$ and every cheap item $i\in I\setminus E$, define
\[
z_{ij}\ =\ \sum_{C\in\mathcal C_q:\ x_C>0} \frac{x_{Cj}}{x_C}\,y_{iC},
\]
and set $z_{ij}=0$ for $j\notin \bigcup_{q\in Q}\mathcal K_q$.

(\emph{Feasibility and objective value.})
We first verify $\sum_{j=1}^m z_{ij}\le 1$ for all items $i\in I$.
For $i\in I\setminus E$,
\[
\sum_{j=1}^m z_{ij}
=
\sum_{q\in Q}\sum_{j\in\mathcal K_q}\sum_{C\in\mathcal C_q:\ x_C>0}\frac{x_{Cj}}{x_C}\,y_{iC}
=
\sum_{q\in Q}\sum_{C\in\mathcal C_q:\ x_C>0}\Bigl(\frac{\sum_{j\in\mathcal K_q}x_{Cj}}{x_C}\Bigr)y_{iC}
=
\sum_{q\in Q}\sum_{C\in\mathcal C_q} y_{iC}
\le 1,
\]
using \eqref{eq:rowsum} and the fifth constraint of the configuration LP.
For $i\in E$, the constraints \eqref{eq:alpha-supply} imply $\sum_{j=1}^m z_{ij}\le 1$ as well.

Fix any knapsack $j\in\mathcal K_q$. We bound its total weight and profit.
By linearity and the definition of $z_{ij}$ for cheap items,
\begin{align*}
\sum_{i\in I\setminus E} p(i)\,z_{ij}
&=
\sum_{C:\ x_C>0}\frac{x_{Cj}}{x_C}\sum_{i\in I\setminus E} p(i)\,y_{iC}
\ge
\sum_{C\in\mathcal C_q}(1-\varepsilon-p(C))\,x_{Cj},\\
\sum_{i\in I\setminus E} w(i)\,z_{ij}
&=
\sum_{C:\ x_C>0}\frac{x_{Cj}}{x_C}\sum_{i\in I\setminus E} w(i)\,y_{iC}
\le
\sum_{C\in\mathcal C_q}\bigl((1+\varepsilon)^q-w(C)\bigr)\,x_{Cj},
\end{align*}
where we used the second and third constraints of the configuration LP.

For expensive items, since each configuration $C\in\mathcal C_q$ specifies the numbers
$\Delta^{u}_{v}(C)=n^u_{\le v}(C)-n^u_{\le(v-1)}(C)$ of items to be taken from each class $E^u_v$,
the definition of $\Delta^{u}_{v}(j)$ and \eqref{eq:alpha-demand} guarantee that knapsack $j$ receives exactly
$\Delta^{u}_{v}(j)$ units from each class $E^u_v$.
(Recall that we work with the rounded instance, hence items within each class $E^u_v$ have identical rounded
profit and weight, so any feasible choice of fractions realizes the same total weight and profit.)
Therefore the total profit/weight of expensive items assigned to $j$ equals
\[
\sum_{i\in E} p(i)\,z_{ij} = \sum_{C\in\mathcal C_q} p(C)\,x_{Cj},
\qquad
\sum_{i\in E} w(i)\,z_{ij} = \sum_{C\in\mathcal C_q} w(C)\,x_{Cj}.
\]
Combining expensive and cheap parts and using \eqref{eq:colsum}, we get
\[
\sum_{i\in I} p(i)\,z_{ij}
\ge
\sum_{C\in\mathcal C_q}\bigl(p(C)+(1-\varepsilon-p(C))\bigr)x_{Cj}
=
(1-\varepsilon)\sum_{C\in\mathcal C_q}x_{Cj}
=
1-\varepsilon,
\]
and
\[
\sum_{i\in I} w(i)\,z_{ij}
\le
\sum_{C\in\mathcal C_q}\bigl(w(C)+((1+\varepsilon)^q-w(C))\bigr)x_{Cj}
=
(1+\varepsilon)^q\sum_{C\in\mathcal C_q}x_{Cj}
=
(1+\varepsilon)^q
=
(1+\varepsilon)B_j.
\]
Hence $\{z_{ij}\}$ is feasible for the assignment LP and has objective value at least $1-\varepsilon$.

(\emph{Property (i).})
By Lemma~\ref{lem:config-sol}, for the smallest $L$ elements $q\in Q$,
the configuration-LP solution assigns expensive items integrally for knapsacks in these $\mathcal K_q$.
In our construction, whenever a knapsack $j\in\mathcal K_q$ receives an integral configuration
(i.e., some $x_{Cj}=1$), we assign whole expensive items to $j$ (by the definition of
$\mathcal J^{\mathrm{int}}_q$ and the integral pre-allocation using $S^{u}_{v}(j)$),
thus $z_{ij}\in\{0,1\}$ for all $i\in E$ and all such $j$.

(\emph{Property (ii).})
Fix $q\in Q$. Since $\{x_{Cj}\}$ is chosen as an extreme point of \eqref{eq:colsum}--\eqref{eq:rowsum},
the number of knapsacks in $\mathcal K_q$ whose column $(x_{Cj})_{C\in\mathcal C_q}$ has at least two
positive entries is at most $|\mathcal C_q|-1$; in particular it is at most $\eta$
because $|\mathcal C_q|\le \eta$.
By construction, only knapsacks in $\mathcal J^{\mathrm{frac}}_q$ may receive fractional expensive items
(they are exactly the knapsacks whose columns are not integral), while for every
$j\in\mathcal J^{\mathrm{int}}_q$ we assign whole expensive items.
Hence at most $\eta$ knapsacks in $\mathcal K_q$ receive fractional expensive items.

(\emph{Property (iii).})
Fix $j\in\mathcal K_q$. Every expensive item assigned to $j$ belongs to some class $E^u_v$ that appears
in a configuration from $\mathcal C_q$ with positive coefficient in $\sum_C x_{Cj}C$.
By the definition of configurations in $\mathcal C_q$, each such expensive item has weight at most
$(1+\varepsilon)^q=B_j$. Therefore $w(i)\le B_j$ for every expensive item $i\in E$ with $z_{ij}>0$.

All steps can be implemented in $2^{2^{poly(1/\varepsilon)}}\cdot poly(|I|)$ time by solving linear programs and performing the
explicit constructions above.
\end{proof}

\end{document}